\documentclass[conference]{IEEEtran}
\IEEEoverridecommandlockouts
\usepackage{cite}
\usepackage{amsmath,amssymb,amsfonts}
\usepackage{algorithmic}
\usepackage{graphicx}
\usepackage{textcomp}
\usepackage{xcolor}
\usepackage{orcidlink}
\usepackage{subfigure}
\usepackage{multirow}
\usepackage{enumitem}
\usepackage{makecell}
\usepackage{booktabs}
\usepackage{amsthm}
\usepackage{caption}
\def\BibTeX{{\rm B\kern-.05em{\sc i\kern-.025em b}\kern-.08em
    T\kern-.1667em\lower.7ex\hbox{E}\kern-.125emX}}

\newtheorem{definition}{Definition}
\newtheorem{theorem}{Theorem}

\begin{document}

\title{From Single to Multiple Attributes: Experimental Insights on Sampling-Based Distinct Combination Estimation in GROUP-BY Queries}
\author{
    \IEEEauthorblockN{Yujie Zhang\orcidlink{0009-0007-3913-6377},Xiaochun Yang\orcidlink{0000-0002-6184-4771}$^*$, Bin Wang\orcidlink{0000-0002-2694-1023},Yuan Sui\orcidlink{0000-0001-7106-6312}}
    \IEEEauthorblockA{Northeastern University, P. R. China. $^*$Corresponding Author}
    \IEEEauthorblockA{2310759@stu.neu.edu.cn,\{yangxc,binwang\}@mail.neu.edu.cn, 2110689@stu.neu.edu.cn}
}

\maketitle

\begin{abstract}
Estimating the number of distinct combinations in multi-attribute GROUP-BY queries remains a significant yet underexplored challenge. Current cardinality estimation techniques primarily focus on SPJ queries (i.e., selections, projections, and joins) and neglect GROUP-BY operations; meanwhile, distinct value estimation research has mainly targeted the single-attribute setting. Although sampling-based methods, including recent approaches with learned models, can theoretically support multi-attribute estimation, their practical effectiveness remains unclear. A comprehensive empirical evaluation is thus lacking to address whether joint distribution information from samples alone is sufficient for accurate multi-attribute estimation, whether existing methods fully exploit single-attribute information and can be further optimized, and whether filtered GROUP-BY queries can be accurately estimated. To this end, we propose a specialized workload generator for multi-attribute GROUP-BY queries and generate both filtered and non-filtered queries over four real-world datasets. By evaluating existing methods across synthetic workloads and the multi-table TPC-H benchmark, we analyze the sources of GROUP-BY cardinality estimation errors and their impact on PostgreSQL's plan selection, offering key recommendations for future estimator design.
\end{abstract}


\section{Introduction}
Estimating the number of distinct values (NDV) is a fundamental problem in many fields, and numerous techniques for single-attribute NDV estimation have been well-established. However, as data complexity grows, real-world queries often involve multiple attributes rather than just one. 
\textcolor{black}{A specific example is Query 16 in the standard industry benchmark TPC-H \cite{tpch}, which performs: \texttt{SELECT ... GROUP BY P\_BRAND, P\_TYPE, P\_SIZE}  to count reliable suppliers for distinct part specifications. This multi-attribute grouping is essential because, in supply chain management, a purchasable part is uniquely defined by the combination of its manufacturer (\texttt{P\_BRAND}), material (\texttt{P\_TYPE}), and dimension (\texttt{P\_SIZE}). This reliance on multi-attribute grouping is even more pronounced in the more advanced TPC-DS benchmark \cite{tpcds},} where $2/3$ of the queries include multi-attribute GROUP-BY operations involving $2$ to $15$ attributes, half of which appear in subqueries.
Accurately estimating the number of distinct multi-attribute combinations is crucial. 
Optimizers rely on precise estimates of distinct combinations to choose appropriate physical operators for aggregation \cite{lee2023analyzing, PostgreSQL, MySQL}, which directly affects the efficiency of the execution plan \cite{PGmulti,lee2023analyzing}. Moreover, when a GROUP-BY operation appears in a subquery, the estimated number of distinct combinations also influences the execution plan of the outer query \cite{freitag2019every}.


\textbf{\underline{Background:}}
Accurate estimation of multi-attribute distinct values is crucial, yet rarely studied explicitly.  First, existing cardinality estimation studies \cite{DeepDB2020,FLAT2021,Naru2019,VarSkip2020,NeuroCard2021,Pessimistic19,wu2022factorjoin,DBLP:journals/pvldb/LiWZDZ023,microsoft20,MSCN2019,DBLP:journals/pvldb/LuKKC21,DBLP:conf/edbt/MengWCZ022,DBLP:conf/sigmod/WuC21,DBLP:journals/pacmmod/Wang0YLMT23,DBLP:conf/sigmod/ZhaoYHLZ22,DBLP:journals/pvldb/LiuD0Z21,DBLP:conf/sigmod/Sun0021,DBLP:journals/corr/abs-2012-14743,DBLP:journals/pvldb/SunZSLT21} focus mainly on SPJ queries, which only consist of selection, projection, and join operations. The goal of these studies is to estimate the total number of tuples that satisfy the queries. To this end, they collapse multiple distinct values during processing \cite{DeepDB2020,FLAT2021,Halford19DASFAA,Naru2019,VarSkip2020,NeuroCard2021,Pessimistic19,wu2022factorjoin}(e.g., via histogram buckets), making them unsuitable for estimating distinct combinations and thus inapplicable to GROUP-BY queries.



Meanwhile, although there have been extensive studies on estimating the number of distinct values \cite{DBLP:journals/pvldb/Ertl24,flajolet2007hyperloglog,charikar2000towards,chao1984nonparametric,chao1992estimating,haas1998estimating,shlosser1981estimation,good1956number,wu2019chebyshev,wu2022learning,li2024learning,DBLP:conf/kdd/LiWDDLZ22,DBLP:journals/corr/abs-2502-16190,DBLP:conf/sigmod/Ting19}, they have mainly focused on single attributes.
Single-attribute NDV estimators are typically classified into sketch-based and sampling-based approaches \cite{Harmouch2017survey,DBLP:conf/edbt/MetwallyAA08}. Sketch-based methods \cite{DBLP:journals/pvldb/Ertl24,DBLP:journals/tkde/WangXZLLLRD24,DBLP:conf/stoc/PettieW21,heule2013hyperloglog,flajolet2007hyperloglog} summarize specified distributions into compact sketches. However, in the multi-attribute setting, they require pre-specified attribute combinations and therefore cannot flexibly adapt to arbitrary user queries and filter predicates. Sampling-based methods \cite{charikar2000towards,chao1984nonparametric,chao1992estimating,haas1998estimating,shlosser1981estimation,good1956number,wu2019chebyshev} derive estimators from frequency information observed in samples. Building on this, recent studies \cite{wu2022learning,li2024learning} further employ learned models as estimation functions. Theoretically, these approaches can support distinct value estimation over an arbitrary number of attributes and filter predicates, but their practical effectiveness in multi-attribute settings has not been systematically evaluated. In addition, Freitag et al.\cite{freitag2019every} leverage sketch-based single-attribute estimates to correct the joint distribution obtained from samples. However, this approach has not been compared against the latest learning-based methods. Consequently, existing research leaves several questions open:

\noindent \textbf{Q1.} Is it sufficient to use only joint distribution information from samples, together with more expressive models (e.g., learned models), to accurately estimate multi-attribute distinct values? 

\noindent \textbf{Q2.} Is there further room for optimizing how existing methods exploit single-attribute information?

\noindent \textbf{Q3.} Are current techniques sufficient to handle more complex scenarios involving filter predicates \textcolor{black}{and joins}?

\textbf{\underline{Contributions and Findings:}}
To address these questions, this paper makes the following contributions:


\textbf{(1) We introduce a general workload generator for evaluating multi-attribute distinct combination estimation.} The generator construct query workloads on four real-world datasets, covering $1$–$68$ GROUP-BY attributes and filtered/unfiltered scenarios, with distinct counts from $1$ to $10^7$.

\textbf{(2) We provide a systematic empirical and theoretical evaluation of traditional and learning-based estimators for multi-attribute distinct estimation.}
Experiments reveal that even advanced learning-based methods relying solely on joint distribution information leads to large errors and biases in multi-attribute settings. Meanwhile, existing boundary-based approaches provide only limited tightening of bounds, leaving room for further optimization.

\textbf{(3) We extend the study to filtered \textcolor{black}{and multi-table TPC-H} GROUP-BY queries, \textcolor{black}{while also analyzing the impact of GROUP-BY cardinality estimation on query planning.}} We show that filter predicates further amplify estimation errors due to reduced effective samples and selectivity bias. \textcolor{black}{We further the practical impact of these errors in PostgreSQL, demonstrating the importance of robust NDV estimation for aggregation strategy selection.}

\textbf{(4) We identify key limitations of current approaches and directions for future optimization.} We identify two future research directions: (i) improving estimation accuracy and (ii) quantifying and reducing uncertainty. We release our code on GitHub \textcolor{black}{\cite{ourcode}} to encourage further studies.

\textbf{\underline{Organization:}}
The rest of the paper is organized as follows: 
Section \ref{sec:2} defines the problem and provides a taxonomy of existing estimators, followed by the experimental setup in Section \ref{sec:Experimental Setup}. Sections \ref{sec:4}--\ref{sec:6} comprehensively evaluate existing estimators for unfiltered and filtered multi-attribute GROUP-BY queries,\textcolor{black}{while Section \ref{sec: tpch} extends this analysis to TPC-H workloads and query plan impacts.}
Future research opportunities are discussed in Section \ref{sec:7}. Finally, we present our conclusions in Section \ref{sec:9}.

\section{From Single to Multiple Attributes: Concepts and Estimators}\label{sec:2}
In this section, we first formulate the distinct combination estimation problem, then organize the sampling-based methods into a taxonomy and analyze how they support multi-attribute distinct value estimation, and finally, we discuss the limitations of existing experimental evaluations.
\vspace{-0.05in}
\subsection{Problem Definition}\label{sec:Problem Definition}

Most existing studies on NDV estimation \cite{DBLP:journals/pvldb/Ertl24,flajolet2007hyperloglog,charikar2000towards,chao1984nonparametric,chao1992estimating,haas1998estimating,shlosser1981estimation,good1956number,wu2019chebyshev,wu2022learning,li2024learning} focus on single attributes, whereas real-world queries often involve multiple attributes and optional selection predicates. Formally:
\begin{definition}[Multi-Attribute Distinct Combination Estimation]\label{def: distinct comb}
Given a relation $\mathcal{R}$, attributes $A_1,\dots,A_k$, and an optional selection predicate $\theta_\mathcal{R}$, let $\sigma_{\theta_\mathcal{R}}(\mathcal{R})$ denote the filtered relation (or $\mathcal{R}$ itself if no predicate). The goal is to estimate the number of unique value combinations:
i.e. $D(A)=\left|\{(v_1,v_2,\dots,v_k) \mid \exists r \in \sigma_{\theta_\mathcal{R}}(\mathcal{R}),r.A_1 =v_1,\dots,r.A_k =v_k \}\right|$.
\end{definition}


Due to correlations between attributes, estimating the cardinality by multiplying the distinct counts of attributes from $A_1$ to $A_k$ may lead to errors of several orders of magnitude \cite{PGmulti2}, which in turn result in suboptimal query plans \cite{PGmulti}. 
The presence of selection predicates further complicates the problem \cite{kipf2019estimating}, as selectivity and attribute correlations can significantly distort value distributions.
\vspace{-0.05in}
\subsection{Taxonomy of Sampling-based Estimators} \label{sec:sampling-based estimators}

In this section, we summarize sampling-based estimators for distinct value estimation.
\textcolor{black}{These methods draw a sample of size $n$, summarize the frequency information in the sample as a \emph{sample profile}, and estimate the number of distinct values $D$ as a function of the sample profile.
For multi-attribute GROUP-BY queries, the notion of \emph{sample profile} naturally extends by treating each distinct attribute-value combination as a single element. We formalize this abstraction as follows.}
\begin{definition}
[\textcolor{black}{Sample Profile for Multi-attribute GROUP-BY}] \label{def:fi}
\textcolor{black}{Consider a GROUP-BY query on attributes \((A_1, \ldots, A_k)\).
Each distinct attribute-value combination \((a_1, \ldots, a_k)\) is treated as one element.
Given a sample \(S\) of tuples from \(R\), let \(f_i\) denote the number of distinct combinations that appear exactly \(i\) times in \(S\).
The vector $\{f_i \mid i \in [1, m] \}$ is called the sample profile. 
An inline example is given below.
Consider a relation \(R(A,B) = \{(1,x), (1,x), (2,y), (2,y), (3,z), (4,z)\}\).
The frequencies of the distinct pairs are
\(\{(1,x):2,\ (2,y):2,\ (3,z):1,\ (4,z):1\}\),
which yields the population frequency of frequency \(\{F_1=2,\ F_2=2\}\).
If a sample \(S = \{(1,x), (1,x), (3,z)\}\) is drawn from \(R\), then the resulting sample profile is \(\{f_1=1,\ f_2=1\}\).}
\end{definition}

Overall, sampling-based estimators aim to infer the number of unseen elements from the observed sample profile. \textcolor{black}{We categorize these estimators into four classes based on how the estimation function is constructed and the statistics they exploit. Table~\ref{tab:ndv_estimators} summarizes these categories and highlights their key properties and evaluation coverage in prior studies.}

\begin{table}[ht]\scriptsize
\centering
\caption{\textcolor{black}{Taxonomy and Comparison of Sampling-based Estimators.}}
\label{tab:ndv_estimators}
\begin{tabular}{l @{\hspace{0.5em}} l @{\hspace{0.5em}} c @{\hspace{0.5em}} c@{\hspace{0.5em}} c @{\hspace{0.5em}}c @{\hspace{0.5em}}c}
\toprule
\textbf{Category}                               &\textbf{Method} & \textbf{Core Dep.}  & \textbf{Linearity} & \textbf{\makecell[l]{Multi Attr. \\Eval.}}  &\textbf{\makecell[l]{Filters.\\ Eval.}} &\textbf{\makecell[l]{Multi Tbl.\\ Eval.}}\\ \midrule
\multirow{2}{*}{\makecell[l]{Singleton-\\aware}}&GEE~\cite{charikar2000towards}      & Singletons          & L & $\times$                   & $\times$& $\times$\\
                                                &Chao~\cite{chao1984nonparametric}            & Singletons          & NL & $\times$& $\times$& $\times$\\\midrule
\multirow{3}{*}{\makecell[l]{Profile-\\aware}}  &Shlosser~\cite{shlosser1981estimation}      & Profile             & NL & $\times$  & $\times$& $\times$\\
                                                &GT~\cite{good1956number}              & Profile             & L& $\times$ & $\times$& $\times$\\
                                                &WY~\cite{wu2019chebyshev}             & Profile             & L  & $\times$& $\times$& $\times$\\\midrule
\multirow{2}{*}{\makecell[l]{Boundary-\\based}} &BC~\cite{freitag2019every}              & Singletons          & NL                 & $\times$                   & $\times$& $\times$\\
                                                &SCBC-T~\cite{freitag2019every}          & \makecell[c]{Singletons+\\Attr. Info}          & NL                 & $\circ$                    & $\times$& $\times$\\\midrule
\multirow{2}{*}{\makecell[l]{Learning-\\based}} &WD~\cite{wu2022learning}              & Profile             & NL                 & $\times$                   & $\times$& $\times$\\
                                                &PolyNet~\cite{li2024learning}         & Profile             & L                 & $\times$                   & $\times$& $\times$ \\ \bottomrule
\addlinespace
\multicolumn{7}{l}{\scriptsize \makecell[l]{\textbf{Legend:} L/NL: Linear/Non-Linear; $\times$: Not evaluated; $\circ$: Evaluated for multi-attri-\\bute queries, but not compared against learning-based models.}} \\
\end{tabular}
\vspace{-0.1in}
\end{table}
\noindent\underline{\textbf{Singleton-aware Estimators:}} Singleton-aware estimators model the unseen portion of the population mainly through low-frequency observations such as $f_1$ and $f_2$.
They assume that high-frequency elements in the population are also likely to appear in the sample and can therefore be estimated as $ \sum^m_{i=2}f_i$. The details of these estimators are as follows:
\begin{enumerate}[leftmargin=12pt,label=\arabic*)]
    \item GEE \cite{charikar2000towards}. 
    GEE estimates the number of singletons as the geometric mean $\sqrt{N/n}\cdot f_1 $ between the lower bound $f_1$ and upper bound $N\cdot f_1/n$.
    \begin{equation} 
    \begin{aligned} \small
    D^{GEE}=\sqrt{N/n}\cdot f_1 + \sum^m_{i=2}f_i=(\sqrt{N/n}-1)f_1+d
    \label{EQ:GEE}
    \end{aligned}
    \end{equation}
    \item Chao \cite{chao1984nonparametric}. Chao \cite{chao1984nonparametric} assumes that the estimated values have a nonlinear relationship with $f_1$ and $f_2$ observed in the sample as $D^{Chao}=d+{f_1}^2/2f_2$. 
    When $f_2=0$, this expression diverges, so we use the bias-corrected version \cite{li2024learning} as Equation (\ref{EQ:Chao}).
    \begin{equation}
    \begin{aligned} \small
    D^{Chao}=\frac{f_1(f_1-1)}{2(f_2+1)}+d
    \label{EQ:Chao}
    \end{aligned}
    \end{equation}
\end{enumerate}

\underline{\textbf{Profile-aware Estimators:}} 
This class of methods leverages the entire sample profile to infer the number of unseen elements. Details are as follows:

\begin{enumerate}[leftmargin=12pt,label=\arabic*)]
    \setcounter{enumi}{2}
    \item Shlosser \cite{shlosser1981estimation}. Shlosser \cite{shlosser1981estimation} was derived under the assumption that the sample frequencies proportionally reflect those in the population, i.e., $E[f_i]/E[f_1] \approx F_i/F_1 $.  
    \begin{equation} \small
    \begin{aligned}
    D^{Shlosser}=\frac{f_1\sum^m_{i=1}(1-\frac{n}{N})^if_i}{\sum^m_{i=1}i\frac{n}{N}(1-\frac{n}{N})^{i-1}f_i}+d
    \label{EQ:Shlosser}
    \end{aligned}
    \end{equation}
     \item GT \cite{good1956number}. GT \cite{good1956number} uses a linear combination of $f_i$ to estimate distinct values as in Equation (\ref{EQ:GT}). 
    \begin{equation} \small
    \begin{aligned}
    D^{GT}=\sum_i(-1)^{i+1}w^if_i+d
    \label{EQ:GT}
    \end{aligned}
    \end{equation}
    \item WY \cite{wu2019chebyshev}. The WY estimator leverages the approximation-theoretic properties of Chebyshev polynomials to propose a linear estimator, as shown in Equation (\ref{EQ:WY}).
    \begin{equation} \small
    \begin{aligned}  
    D^{WY}=\sum_{i=1}^{L}g_L(i)f_i+ {\sum_{i>L}f_i} 
    \label{EQ:WY}
    \end{aligned}
    \end{equation}
    where $g_L$ is a polynomial of degree $L$, as in GT, it is also oscillating and results in coefficients with alternating signs.

\end{enumerate}
\underline{\textbf{Boundary-based Estimators:}} This class of method theoretically derives the bounds for singletons. Specifically,
\begin{enumerate}[leftmargin=12pt,label=\arabic*)]
\setcounter{enumi}{5}
\item BC \cite{freitag2019every}. The BC (Bound-Corrected) estimator \cite{freitag2019every} establishes a relationship between the expected number of distinct values $D$ in the full table, the observed distinct count $d$, and singleton count $f_1$ in the sample, yielding tighter singleton bounds $L_{BC}$ and $U_{BC}$ than those in GEE, as shown in Equation (\ref{EQ:BC}).
    \begin{equation}\small
    \begin{aligned}
    D^{BC} = \sqrt{ \hat{L}_{BC} \cdot \hat{U}_{BC} } +\sum^m_{i=2}f_i
    \label{EQ:BC}
    \end{aligned}
    \end{equation}
\item SCBC-T \cite{freitag2019every}. The SCBC estimator (Sketch-Corrected BC)~\cite{freitag2019every} replaces $L_{BC}$ and $U_{BC}$ with tighter bounds $L_{SCBC}$ and $U_{SCBC}$ that incorporate single-attribute information: $\hat{L}_{SCBC} = \max( \hat{L}_{BC}, \max(\hat{F_{1,j}}) _{j=1,\ldots, C} )$, $\hat{U}_{SCBC} = \min ( \hat{U}_{BC}, \prod_{j=1}^{C}\hat{D_j})$, where $\hat{F_{1,j}}$ and $\hat{D_j}$ are estimated using single-attribute sketches. In this paper, we replace these estimates with their true values $F_{1,j}$ and $D_j$  to eliminate sketch-induced errors, and refer to the resulting upper-bound accuracy as SCBC-T.
\end{enumerate}
\underline{\textbf{Learning-based Estimators:}}
\textcolor{black}{This class of methods remains sampling-based, as the estimators operate exclusively on the sample profile. The key distinction is that the estimation function is learned from data rather than analytically derived.}
\begin{enumerate}[leftmargin=12pt,label=\arabic*)] 
    \setcounter{enumi}{7}
    \item WD \cite{wu2022learning}.
    Wu et al. \cite{wu2022learning} model the NDV estimation problem as a maximum likelihood estimation (MLE) problem, aiming to find the NDV value that maximizes the probability of observing the sample profile. The MLE problem is further transformed into a supervised learning problem, where the sample profile is used as features and the NDV as labels to train a neural network, as shown in Equation (\ref{EQ:WD}).
    \begin{equation} \small
    \begin{aligned}
    D^{WD}=h(f_1,\dots,f_m,N)
    \label{EQ:WD}
    \end{aligned}
    \end{equation}
    \item PolyNet \cite{li2024learning}. Li et al. combined the ideas of WY and WD, using a simple two-layer linear activation network called PolyNet to estimate the parameters of the polynomial. 
\end{enumerate}

\vspace{-0.03in}
\subsection{Limitations of Existing Experiments} \label{sec:Limitations of Existing Experiments}

\textcolor{black}{Table~\ref{tab:ndv_estimators} provides a structured comparison of representative sampling-based estimators and summarizes the key aspects that have (or have not) been evaluated in prior work. Based on this comparison, we elaborate on the major limitations below.}

\textcolor{black}{\textbf{Evaluation Gap in Multi-Attribute Scenarios.} By treating each unique combination $(a_1, \ldots, a_k)$ as a single element, any estimator designed for single-attribute NDV can, in principle, be applied to multi-attribute GROUP-BY queries. However, their performance in multi-attribute settings remains largely unvalidated. 
As shown in Table~\ref{tab:ndv_estimators}, most methods in the table, as well as existing surveys~\cite{liu2020sampling,Harmouch2017survey,abedjan2015profiling,DBLP:journals/dase/LanBP21}, focus on single columns and do not evaluate multi-attribute scenarios.}

\textcolor{black}{\textbf{Isolated Development of Boundary-based and Learning-based Methods.} The SCBC framework~\cite{freitag2019every} adopts a hybrid design that leverages auxiliary single-attribute information to tighten bounds for multi-attribute distinct estimation. However, it has not been systematically compared with state-of-the-art learning-based approaches. As a result, it remains unclear whether the expressive capacity of neural models, such as WD and PolyNet, can outperform the theoretical guarantees of boundary-based estimators like SCBC-T.}

\textcolor{black}{\textbf{Lack of Evaluation under Selection Predicates and Multi-Table Queries.} Theoretically, sampling-based estimators are applicable to arbitrary predicates, as they estimate using only the filtered samples. For multi-table queries, they can also be applied in principle as long as join samples are available. However, such scenarios have largely been unexplored in prior studies, as shown in Table \ref{tab:ndv_estimators}.}



\begin{table*}[ht]\scriptsize
\centering
\caption{Description of workloads.}
\label{tbl:workload}
\begin{tabular}{c|c@{\hspace{1em}} c@{\hspace{1em}} c @{\hspace{1em}}c @{\hspace{1em}}c| c@{\hspace{1em}}  c@{\hspace{1em}}  c @{\hspace{1em}} c |c @{\hspace{1em}} c @{\hspace{1em}} c @{\hspace{1em}} c} 
\toprule
\multirow{2}{*}{} & \multicolumn{5}{c|}{Data} &\multicolumn{4}{c|}{Non-filter Workload} &\multicolumn{4}{c}{Filtered Workload}\\
&Rows&Col/Cat&Domain&\textcolor{black}{Skewness}&\textcolor{black}{Variance}&Queries&Cols &Single NDV&Multi NDV&Queries&Single NDV&Multi NDV&\textcolor{black}{Selectivity}\\ 
\midrule
Census  &2.5M  &68/68 &$10^{51}$&\textcolor{black}{398-2.4M}&\textcolor{black}{$10^9-10^{12}$}&3351&$1-68$&$2-16$         &$2-10^6$ &$10050$&$1-16$&$2-10^6$&\textcolor{black}{$10^{-5}-1$}\\ 
Airline &10.0M &10/3  &$10^{22}$&\textcolor{black}{1-1.9M}&\textcolor{black}{$10^7-10^{11}$}&1023&$1-10$&$7-1786$       &$84-10^7$&$5035$ &$1-1769$&$1-10^7$&\textcolor{black}{$10^{-7}-1$}\\ 
DMV     &12.2M &20/12 &$10^{51}$&\textcolor{black}{1-12.1M}&\textcolor{black}{$10^6-10^{13}$}&1319&$1-19$&$2-10^4$&$4-10^7$ &$3954$ &$1-10^4$&$1-10^7$&\textcolor{black}{$10^{-8}-1$}\\ 
Campaign&4.3M  &21/17 &$10^{73}$&\textcolor{black}{1-4.3M}&\textcolor{black}{$10^{-2}-10^{12}$}&1319&$1-19$&$3-10^6$&$18-10^6$&$3850$ &$1-10^5$&$1-10^6$&\textcolor{black}{$10^{-7}-1$}\\ 
\bottomrule
\addlinespace
\multicolumn{14}{l}{\scriptsize \makecell[l]{\textbf{Legend:}  Cols/Cat: the total number of columns and categorical columns; Domain: the product of the distinct value counts across all columns; Skewness: the min and max fre-\\quencies for the attribute with the largest frequency gap; Variance: the range of frequency variances across attributes; Single/Multi NDV : the query cardinality ranges for single\\attribute/multi-attribute GROUP-BY queries, respectively.}} \\
\end{tabular}
\vspace{-0.1in}
\end{table*}
\section{Experimental Setup}\label{sec:Experimental Setup}
\vspace{-0.05in}
This section provides a detailed description of the general setup used in all our experiments.

\underline{\textbf{Environment.}} All experiments are performed on a Linux machine with 12th Gen Intel(R) Core(TM) i7-12700F CPU, 128 GB DDR4 main memory and 2 NVIDIA Tesla P100 GPU.

\underline{\textbf{\textcolor{black}{Implementation and Parameters.}}} \textcolor{black}{To ensure a fair comparison, all methods are evaluated using the open‑source Python implementations provided in prior works~\cite{li2024learning,wu2022learning}. For BC and SCBC \cite{freitag2019every}, which do not have public repositories, we faithfully re-implemented them in Python based on their original algorithmic descriptions. The hyperparameters for learning-based models (WD and PolyNet) and the configuration for statistical methods are kept strictly consistent with those reported in the original research.}

\underline{\textbf{\textcolor{black}{Workload and Dataset.}}}
\textcolor{black}{As shown in Table \ref{tab:ndv_estimators}, original evaluations typically focused on estimating the NDV of isolated columns within a dataset. In contrast, we developed a Workload Generator to simulate real-world SQL scenarios}:
\begin{enumerate}[leftmargin=12pt,label=\arabic*)]
\item \textcolor{black}{\textbf{Multi-attribute GROUP-BY Queries:} Instead of predefined columns, our generator enumerates and samples attribute combinations (from 2-column to full-row keys), excluding primary keys to ensure the queries are meaningful.} We generated workloads on four real-world datasets, Census \cite{census}, Airline \cite{airline}, DMV \cite{dmv} and Campaign \cite{campaign}. They are commonly used in related research \cite{wu2022learning,li2024learning,freitag2019every,DBLP:journals/pvldb/WangQWWZ21}, and vary in column count and domain size, with skewed data distributions. Details are provided in Table \ref{tbl:workload}.
\item \textcolor{black}{\textbf{Filtered GROUP-BY Queries:} 
Based on the non-filter workload, we further introduce predicates on non-grouping attributes.} We focus on range predicates, as the effects of selectivity and correlation are independent of the specific predicate type. For each query, three filtered variants with different selectivity levels are generated. The maximum selectivity is close to but less than $1$. Details of the filter workload and selectivity ranges are provided in Table~\ref{tbl:workload}.
\item  \textcolor{black}{\textbf{Cross-table Evaluation (TPC-H):} Beyond the workloads generated by our workload generator, we also evaluate multi-table GROUP-BY queries on the industrial TPC-H benchmark, considering GROUP-BY over joined relations.} 
\end{enumerate}

\underline{\textbf{Evaluation Metric.}}
\textcolor{black}{
We evaluate the methods from two aspects: \textit {accuracy} and \textit {latency}, reflecting both result quality and practical efficiency for the GROUP-BY distinct computation.}
\begin{itemize}  [topsep=0pt,leftmargin=12pt]
    \item \textit {Accuracy.} We adopt two standard metrics to quantify estimation accuracy : $Q-error=\max(\frac{Estimated}{True},\frac{True}{Estimated})$ measures the multiplicative deviation and symmetrically penalizes over- and under-estimation. A value of $1$ indicates perfect estimation. Rel-error, defined as $\frac{Estimated}{True}$, which captures the directional bias of estimation. Values greater (less) than $1$ indicate overestimation (underestimation). These metrics complement each other: Q-error summarizes overall accuracy, while Rel-error highlights trends of over- or under-estimation.
    \item \textcolor{black}{\textit{Latency.} We report latency of distinct estimation by separating it into two components: (i) the sampling time, and (ii) the inference time for estimating the number of distinct combinations. }

\end{itemize}
\vspace{-0.08in}
\begin{table*}[ht]\scriptsize
\centering
\caption{Average Q-error of methods with different sampling rates on single-attribute queries.}
\label{tbl:1}
\vspace{-0.1in}
\begin{tabular}{c| c c c c |c c c c |c c c c |c c c c} 
\toprule
\multirow{2}{*}{Method} & \multicolumn{4}{c|}{Census} &\multicolumn{4}{c|}{Airline} &\multicolumn{4}{c|}{DMV}&\multicolumn{4}{c}{Campaign}\\
& 0.001&0.005 & 0.007 & 0.01 & 0.001& 0.005 &0.007 &0.01 & 0.001&0.005 & 0.007 & 0.01 & 0.001 & 0.005 & 0.007 & 0.01  \\ 
\midrule
GEE      & 1.470& 1.068 & 1.045 & 1.053               & 2.803& 1.392& 1.272  &1.231&                     \textbf{2.061}& 1.456& 1.453 & 1.397& 5.028&  3.088 & 2.593 &  2.375  \\
Chao     & \textbf{1.024}& \textbf{1.006} & 1.012 & \textbf{1.002}     & 1.430& 1.248& 1.188  &1.188&                    3.848& 2.293& 2.070 &  1.846& 8.098&  4.719& 3.944 &  4.481  \\
Shlosser & 1.067& 1.020 & 1.023  & 1.019               & 5.646& 1.170& \textbf{1.081 }&\textbf{1.069} & 29.34& 1.620& 1.496 &  1.412& 4.552& 4.428 & 3.148 &  2.125 \\
GT       & 1.027& \textbf{1.006} & \textbf{1.011}  & \textbf{1.002}               & 1.610& 1.342& 1.285  &1.256&                    6.013& 3.269& 2.959 &  2.626& 114.6& 26.11& 19.28 & 14.24 \\
WY       & 1.596& 1.307 & 1.335  & 1.161              & 2.145& 1.295& 1.657  &1.201&                     5.836& 3.928& 2.324 &  3.777& 14.49&  4.256& 3.318 &  3.315 \\
BC       & 1.497& 1.063 & 1.042  & 1.048               &2.739 & 1.354& 1.243  & 1.196 & 2.102& 1.437  & \textbf{1.427} & 1.368 & \textbf{2.147} & \textbf{1.920} & \textbf{1.611}  & \textbf{1.581}\\
SCBC-T     & 1.015& 1.004 & 1.009 & 1.007 &1.088& 1.065& 1.066& 1.072 & 1.398 &1.230& 1.196& 1.173 & 1.744 & 1.510 &1.508&1.475\\
WD       & 1.334& 1.257 & 1.232  & 1.198               & 1.416& 1.260& 1.231  &1.183& 2.508& 1.637 & 1.555 & 1.393&  3.357&  2.238& 2.021 &  1.815  \\
PolyNet  & 1.028& 1.120 & 1.308  & 1.051               & \textbf{1.193}& \textbf{1.057}& 1.217  &1.076&   3.993& \textbf{1.372}& 1.561 &  \textbf{1.333} & 4.252&  3.246& 2.433 & 1.797 \\
\bottomrule

\end{tabular}
\end{table*}

\begin{table*}[ht]\scriptsize
\centering
\caption{Average Q-error of methods with different sampling rates on multi-attribute queries.}
\label{tbl:n}
\vspace{-0.1in}
\begin{tabular}{c| c c c c |c c c c |c c c c |c c c c} 
\toprule
\multirow{2}{*}{Method} & \multicolumn{4}{c|}{Census} &\multicolumn{4}{c|}{Airline} &\multicolumn{4}{c|}{DMV}&\multicolumn{4}{c}{Campaign}\\
& 0.001&0.005 & 0.007 & 0.01 & 0.001&0.005 & 0.007 & 0.01 & 0.001&0.005&0.007 & 0.01 & 0.001&0.005 & 0.007 & 0.01\\ 
\midrule
GEE      & 14.50 & 6.963 & 5.963 &  5.068     &18.96& 8.346 & 7.021 &  5.844 &19.29& 8.636 & 7.300 & 6.108  & 24.60 &10.74 & 9.002 &  7.464\\
Chao     & 21.62 & 6.817 & 5.448 &  4.236     & 3.852&14.99 & 18.82 & 27.53 & 7.636& 3.946 & 3.512 & 3.509  & 4.842& 29.19 & 54.58&  108.7\\
Shlosser & \textbf{2.556}  & \textbf{1.543}& \textbf{1.439}&   \textbf{1.341}  & 6.884& 2.973 & 2.562 &  2.228 & 2.873& 1.657 & 1.534 & 1.438 & 1.723&  1.378 &1.318 & 1.268\\
GT       & 408.2&88.60 & 64.62 & 46.56       &557.5&111.7 & 79.97 &56.21  &569.6&115.3 & 82.92 &58.60  &723.7&146.4 & 104.9 & 73.68\\
WY       & 45.78 &10.70 & 8.060 &  5.967     &47.06&10.02 & 7.368 &  5.359  &49.51& 10.85 & 8.078 &5.877  &77.54&  16.21 & 11.770 & 8.449\\
BC       & 4.581& 2.969 & 2.679 & 2.414      & 5.248&1.638 & 1.332 & 1.257 & 1.789 & 1.447 & 1.390 & 1.340 &  1.437 & 4.918 & 4.329 & 1.218\\
SCBC-T     & 4.454& 2.930 & 2.653 & 2.396      & 5.031&1.543 & \textbf{1.257} & 1.332 & \textbf{1.747} & \textbf{1.429} & \textbf{1.371} & \textbf{1.321} & \textbf{1.224} & \textbf{1.196} & \textbf{1.183} & \textbf{1.170} \\
WD       & 6.338  & 3.561 & 3.253 &  2.891    &\textbf{1.493}& \textbf{1.309} & 1.273 &  \textbf{1.250}  & 3.004& 2.096 & 1.984 & 1.806  & 2.017&  1.527 & 1.457 &  1.383\\
PolyNet  & 18.45 & 4.172 & 3.129 &  2.348    &11.34& 5.167 & 3.879 &  2.949  &23.57& 5.190 & 3.824 & 2.833 & 7.322&  6.374 & 4.591 &  3.264\\
\bottomrule
\end{tabular}
\end{table*}
\section{Can Single-Attribute Methods Handle Multi-Attribute Scenarios?}\label{sec:4}
The sample profile described in Definition~\ref{def:fi} captures only joint information; is this sufficient for accurate estimation?
This section investigates this issue for multi-attribute GROUP-BY queries without filters: Section \ref{sec:4.1} compares the estimation accuracy of single-attribute and multi-attribute GROUP-BY queries under different sampling rates, Section~\ref{sec: theo} theoretically analyzes the potential benefit of incorporating single-attribute information, and Section~\ref{sec:4.3} analyzes the computational overhead.

\vspace{-0.08in}
\subsection{Are Sample Profiles Effective for Distinct Combination Estimation?}\label{sec:4.1}
In this section, we evaluate overall accuracy on the non-filter workload. Following prior work~\cite{wu2022learning,li2024learning}, we use sampling rates from $0.001$ to $0.01$, repeat each experiment five times, and report average Q-errors for single- and multi-attribute GROUP-BY queries in Tables~\ref{tbl:1} and~\ref{tbl:n}. Since SCBC-T already uses single-attribute NDV information, its single-attribute accuracy is excluded from comparison; thus, Table~\ref{tbl:1} highlights the best results among the remaining methods.
In the following, we present the general observations (denoted as \textbf{O}), followed by a detailed analysis.

\textbf{O1: Relying solely on the joint distribution information, even state-of-the-art learning-based models exhibit larger errors in multi-attribute scenarios compared to single-attribute cases.}
Almost all methods achieve near-optimal accuracy (Q-error close to $1$) for single-attribute queries, but for multi-attribute queries, sample-profile-only methods show much larger errors. For instance, on the Census dataset with a $0.001$ sampling rate, multi-attribute GROUP-BY errors are $2$–$400\times$ higher than single-attribute ones, and even at $0.01$ sampling rate, errors remain up to $46\times$ larger.
These errors stem from the exponential growth of attribute combinations, which challenges \textcolor{black}{even SOTA learning-based methods like WD and PolyNet} when relying solely on joint samples. Consequently, incorporating single-attribute information becomes essential. \textcolor{black}{The Attr. Info-aware method SCBC-T} shows clear advantages on datasets with higher NDV (e.g., DMV and Campaign).
However, this advantage vanishes on low-NDV datasets such as Census, since the bounds derived by SCBC-T using single-attribute information typically reduce to BC. Section~\ref{sec:5} provides a detailed analysis.



\begin{figure*}[ht]
\centering
\begin{minipage}[]{0.8\textwidth}
\includegraphics[width=\linewidth]{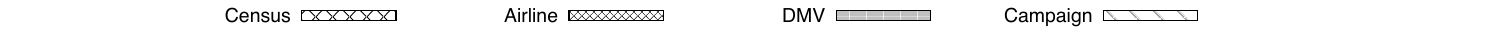}
\end{minipage}\\
\vspace{-0.05in}
\subfigure[\label{fig: samplingtime} \scriptsize \textcolor{black}{Sampling time}]{\includegraphics[width=0.18\textwidth]{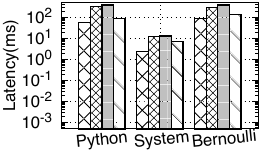}}
\subfigure[\label{fig: latency1} \scriptsize Inference latency on single-attribute queries]{\includegraphics[width=0.4\textwidth]{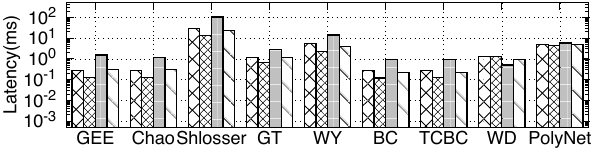}} 
\subfigure[\label{fig: latency2}\scriptsize Inference latency on multi-attribute queries]{\includegraphics[width=0.4\textwidth]{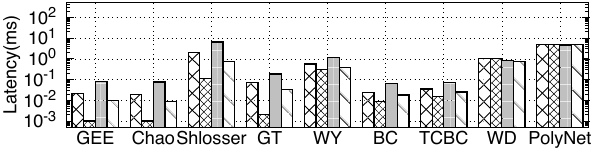}} 
\caption{Inference overhead on single- and multi-attribute queries.}
\label{fig: latency}
\vspace{-0.15in}
\end{figure*}
\vspace{-0.05in}
\subsection{Theoretical Analysis of Integrating Single-Attr. Information}\label{sec: theo}


\textcolor{black}{Prior negative results for sampling-based distinct estimation \cite{wu2022learning, charikar2000towards} establish that datasets sharing a common subset $T_0$ can remain statistically indistinguishable despite having a distinct count gap of $k$, implying unavoidable worst-case errors. In the absence of single-attribute information, this gap $k$ can be arbitrarily large as long as $T_0$ is captured by the sample. However, our empirical findings in Section \ref{sec:4.1} demonstrate that methods integrating single-attribute statistics, such as SCBC-T, exhibit significantly better stability. This improvement is theoretically grounded in the fact that observing marginal distributions bounds the total NDV by the Cartesian product of the marginal domains, effectively constraining the feasible range of $k$. To formalize this, the following theorem derives a tighter lower bound on estimation error when marginal distributions are known.}

\begin{theorem}
For a dataset of $N$ rows with two categorical attributes, let $\hat{D}$ be an estimator using a sample of size $n$ and known marginal distributions. For any $\gamma > e^{-4n}$, there exists a dataset such that with probability at least $\gamma$:
\begin{equation} \small
error(\hat{D}) \ge \left( \frac{N - n}{4n} \ln \frac{1}{\gamma} \right)^{\frac{1}{4}}
\end{equation}
\label{thm::theorem1}
\end{theorem}
\vspace{-0.15in}
\begin{proof}
\vspace{-0.15in}
Following the strategy in \cite{charikar2000towards,wu2022learning}, we construct two datasets with identical marginal distributions over attributes $X$ and $Y$, but with vastly different numbers of distinct $(X, Y)$ pairs. With probability at least $\gamma$, the same sample may be drawn from both. Specifically, given the marginal distributions of $X$ and $Y$, the maximum distinct $(X, Y)$ pairs is the Cartesian product, minimized under perfect bijection. So we construct two datasets $T_1 = T_0 \cup \Delta_1$ (Cartesian case) and $T_2 = T_0 \cup \Delta_2$ (bijective case). Let $X$ and $Y$ each have $d$ distinct values; then $T_1$ has $d^2$ distinct $(X, Y)$ pairs, while $T_2$ has only $d$. The shared part $T_0$ contains bijective pairs $(X_i, Y_i)$ for $i = 1, \dots, d$. $\Delta_1$ includes the off-diagonal Cartesian pairs $(X_i, Y_j)$ for $i \ne j$, each appearing once. Hence, $\Delta_1$ contains $(d^2 - d)$ tuples. $\Delta_2$ is of the same size as $\Delta_1$, but contains only values that already have appeared in $T_0$. A concrete example illustrates this construction: 
\begin{equation}  \small
\begin{aligned}
&T_0 = \{(1,A),(2,B),(3,C)\}\\
&\Delta_1 = \{(1,B),(1,C),(2,A),(2,C),(3,A),(3,B)\}\\
&\Delta_2 = \{(1,A),(2,B),(3,C),(1,A),(2,B),(3,C)\}\\
\end{aligned}
\end{equation}

We now proceed to analyze the probability that a uniformly random sample $S_1 = \{a_1, \dots, a_n\}$, drawn from $T_1$, is entirely contained within $T_0$. Note that the size of $T_0$ is $N - |\Delta_1| = N - (d^2 - d)$, when drawing the $i$-th value $a_i$,
\begin{equation} \small
\begin{aligned}
P(a_i\in T_0 \mid a_1 \in T_0,...,a_{i-1}\in T_0)=\frac{N-(d^2-d)-i+1}{N-i+1} 
\end{aligned}
\end{equation}
thus, the probability that $S_1$ contain values only from $T_0$ is
\begin{equation} \small 
\begin{aligned}
P(S_1 \subset T_0) &= \prod_{i=1}^{n} \frac{N-(d^2-d)-i+1}{N-i+1} \ge \left(\frac{N-n-(d^2-d)}{N-n} \right)^n \\
&\ge \left(\frac{N-n-d^2}{N-n} \right)^n 
= \left(1-\frac{d^2}{N-n} \right)^n 
\label{eq: s1 in t0}
\end{aligned}
\end{equation}

Similarly, the probability that a sample drawn from $T_2$ is entirely contained within $T_0$ is $P(S_1 \subset  T_0) \ge \left(1-\frac{d^2}{N-n} \right)^n$. When $S_1 \subset T_0$ and $S_2 \subset T_0$, the sample profiles $f_1$(from $S_1$) and $f_2$(from $S_2$) follow identical distributions, formally, $P(f_1 = f \mid S_1 \subset  T_0 \wedge  S_2 \subset  T_0) = P(f_2 = f \mid S_1 \subset  T_0 \wedge  S_2 \subset  T_0)$. Thus, we cannot determine whether the underlying dataset is $T_1$ or $T_2$ based on the observed sample profile. From Equation (\ref{eq: s1 in t0}), we have 
\begin{equation} \small
\begin{aligned}
P(S_1 \subset  T_0 \wedge  S_2 \subset  T_0) \ge \left(1-\frac{d^2}{N-n} \right)^{2n}  \ge e^{\frac{-4d^2n}{N-n}}
\end{aligned}
\end{equation}
where the last inequality follows from the claim that $1 - z \ge e^{-2z}$ for $z = d^2/(N - n) \le 1/2$, which can be easily verified \cite{charikar2000towards,wu2022learning}. Let $\gamma=e^{\frac{-4d^2n}{N-n}}$ and thus $d=\sqrt{\frac{N-n}{4n}ln\frac{1}{\gamma }} $. 

Thus, for $\gamma \ge e^{-4n}$ (which ensures $|T_0| = N - (d^2-d)\ge N - d^2 \ge n$, i.e., a sufficient number of values can be drawn from $T_0$ into $S_1$ or $S_2$) and $z = d^2/(N-n) \le 1/2$, with probability at least $\gamma$, we cannot distinguish whether the sample profile $f$ is from $T_1$ or $T_2$. For estimation $\hat{D}$, the error is at least $\max(\frac{D_1}{\hat{D}},\frac{\hat{D}}{D_2})$, which is minimized when $\hat{D}=\sqrt{D_1D_2}=\sqrt{d^3}$. And thus
\begin{equation} \small
\begin{aligned}
error(\hat{D})\ge \frac{\sqrt{d^3}}{d}=\sqrt{d}=\left(\frac{N-n}{4n}ln\frac{1}{\gamma }\right)^{\frac{1}{4}}
\end{aligned}
\end{equation}
\end{proof}

\vspace{-0.12in}



\textcolor{black}{In summary, Theorem~\ref{thm::theorem1} demonstrates that the presence of marginal information tightens the lower bound of the worst-case error compared to the bound established in the absence of such information, i.e., $\sqrt{\frac{N-n}{2n}\ln\frac{1}{\gamma}}$~\cite{charikar2000towards}. This confirms theoretically that supplementing boundary information can prune the space of possible joint distributions, which aligns with our experimental observation: the boundary-based method SCBC‑T achieves relatively more stable error in multi-attribute scenarios. However, SCBC‑T frequently degenerates to BC, indicating that the bounds derived from the marginals are often too loose to be effective. This gap highlights substantial room for future research on how to more effectively integrate marginal statistics, beyond simple Cartesian-product or singleton bounds, into multi-attribute distinct estimation.}

\begin{figure*}[ht!]
\centering
\vspace{-0.05in}
\subfigure[\label{fig: Census-1}\scriptsize Census]{\includegraphics[width=0.245\textwidth]{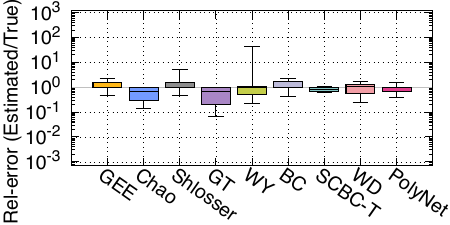}} 
\subfigure[\label{fig: Airline-1}\scriptsize Airline]{\includegraphics[width=0.245\textwidth]{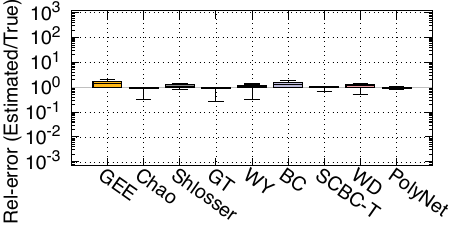}}
\subfigure[\label{fig: DMV-1}\scriptsize DMV]{\includegraphics[width=0.245\textwidth]{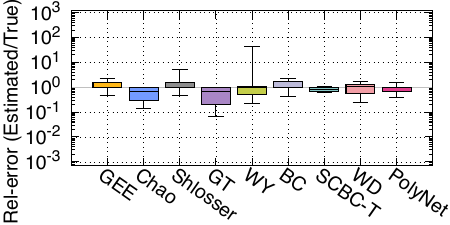}}
\subfigure[\label{fig: Campaign-1}\scriptsize Campaign]{\includegraphics[width=0.245\textwidth]{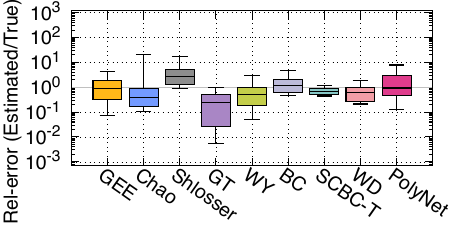}} 
\vspace{-0.15in}
\caption{\label{fig: single}Rel-error distribution on single-attribute GROUP-BY queries.}
\vspace{-0.05in}
\end{figure*}
\begin{figure*}[ht!]
\centering
\vspace{-0.05in}
\subfigure[\label{fig: Census-n}\scriptsize Census]{\includegraphics[width=0.245\textwidth]{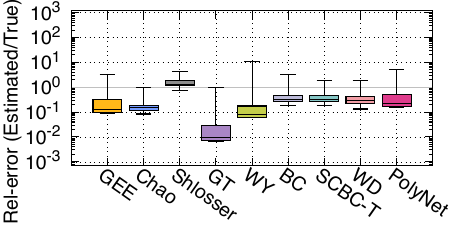}} 
\subfigure[\label{fig: Airline-n}\scriptsize Airline]{\includegraphics[width=0.245\textwidth]{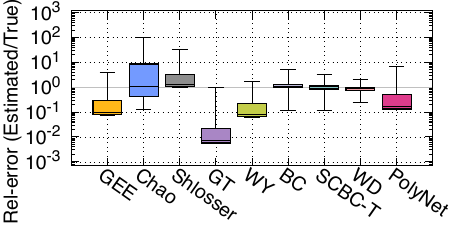}}
\subfigure[\label{fig: DMV-n}\scriptsize DMV]{\includegraphics[width=0.245\textwidth]{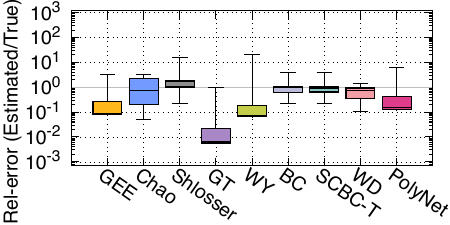}}
\subfigure[\label{fig: Campaign-n}\scriptsize Campaign]{\includegraphics[width=0.245\textwidth]{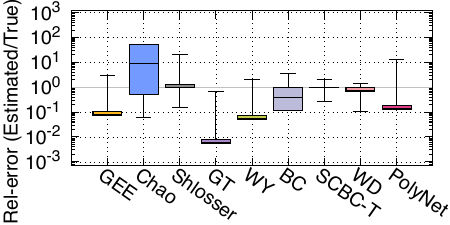}} 
\vspace{-0.15in}
\caption{\label{fig: multiple}Rel-error distribution on multi-attribute GROUP-BY queries.}
\vspace{-0.05in}
\end{figure*}

\subsection{What About Inference and Deployment Cost?} \label{sec:4.3}
In this section, we evaluate the computational overhead of existing methods for estimating single-attribute and multi-attribute GROUP-BY queries. 
\textcolor{black}{Figure \ref{fig: samplingtime} summarizes the average sampling overhead. To rigorously assess the data access overhead in a real-world system, we additionally measure sampling time in PostgreSQL using built-in sampling operators, including both \texttt{TABLESAMPLE SYSTEM} (page-level sampling) and \texttt{TABLESAMPLE BERNOULLI} (tuple-level sampling) overhead.}
Figures~\ref{fig: latency1} and \ref{fig: latency2} show the estimation time of each method given a sample. Since SCBC-T serves as a theoretical baseline that leverages single-attribute information, Figure \ref{fig: latency2} reports the estimation time given single-attribute information.
The key observations are summarized as follows:

\textbf{O2: For multi-attribute distinct combination estimation, all methods incur relatively low inference costs. For deployment, learning-based methods involve additional training overhead; but their integration with sampling provides strong generalization capability, enabling a single model to adapt to varying data distributions.}

\textcolor{black}{Overall, as illustrated in Figure \ref{fig: samplingtime}, sampling latency scales with dataset size across all implementations. 
The Python-based implementation ($58$–$391$ ms) and PostgreSQL’s \texttt{BERNOULLI} ($86$–$385$ ms) exhibit comparable overheads, as both require a full scan for tuple-level random sampling, with the latter also including additional data access costs.
For latency-critical applications, PostgreSQL’s \texttt{SYSTEM} significantly reduces I/O costs by sampling data pages rather than individual tuples, completing within a mere $2.4$ to $13$ ms. Ultimately, in practical query optimization, samples are typically pre-computed and cached as part of background statistics maintenance (e.g., the \texttt{ANALYZE} mechanism in PostgreSQL). Consequently, the sampling process does not impose any runtime overhead during the active query inference phase.} 

For inference cost, Figures~\ref{fig: latency1} and \ref{fig: latency2} show that non-learning-based methods infer multi-attribute queries faster than single-attribute ones, since their inference cost depends on the sample profile size. Multi-attribute queries yield highly sparse profiles (most values appear once), reducing profile size and inference time. In contrast, learning-based methods pad profiles to fixed-length vectors, resulting in stable inference time regardless of the number of attributes.

For deployment, learning-based methods incur additional training costs: WD and PolyNet require $7948$s and $421$s for training, respectively, and generating the shared 100K synthetic training samples takes over 6 hours. However, once trained, these models are reusable across datasets and distributions, unlike prior non-sampling-based learning-based SPJ estimators that require retraining per dataset \cite{DeepDB2020,FLAT2021,Naru2019,VarSkip2020,NeuroCard2021,microsoft20,MSCN2019}. This reuse highlights the generalization enabled by sampling and provides a stable cost profile for multi-attribute distinct estimation, allowing future work to focus on improving accuracy.

\section{Where Are the Opportunities for Improvement?}\label{sec:5}
In this section, we conduct an in-depth analysis of the sources of estimation errors and explore potential directions for optimization. Specifically, Section \ref{sec:5.1} provides a overall analysis of the error distributions of the various methods, while Section \ref{sec:5.2} conducts a query-level analysis to examine the sources of error.

\vspace{-0.05in}
\subsection{Overall Error Trends}\label{sec:5.1}
\vspace{-0.05in}
This section compares error distributions across all methods. \autoref{fig: single} and \autoref{fig: multiple} show the relative error distributions for single- and multi-attribute queries at a 0.5\% sampling rate. In each boxplot, the five vertical lines represent the minimum, 25th percentile, median, 75th percentile, and maximum errors. Results closer to $1(10^0)$ indicate higher accuracy, while values above or below $1$ signify overestimation and underestimation, respectively. Key observations are summarized below.


\textbf{O3: Methods based on heuristic distributions show systematic over- or under-estimation, even with learned parameterization.}
As shown in Figures \ref{fig: single} and \ref{fig: multiple}, when handling multiple attributes, the boundary-based methods (BC and SCBC-T) and the learning-based estimator WD exhibit a more stable error distribution. 
\textcolor{black}{In contrast, estimators that rely on strong distributional assumptions show systematic biases. Methods that assume a linear relationship between the number of distinct values and the observed frequencies (including GEE, GT, WY, and PolyNet) tend to underestimate the true distinct count, whereas nonlinear estimators (Chao and Shlosser) consistently overestimate.} Notably, while both WY and PolyNet are rooted in polynomial approximation, PolyNet’s use of a neural network to learn coefficients allows it to consistently outperform WY, demonstrating the benefit of incorporating learning-based parameterization. However, due to its underlying linear structure, PolyNet remains less accurate than the MLE-based learning model WD in multi-attribute settings.
These systematic under- and over-estimation patterns are primarily driven by the singleton count and are further analyzed in the query-level error analysis (Section~\ref{sec:5.2}).

\begin{figure*}[ht!]
\centering
\begin{minipage}[]{0.8\textwidth}
\includegraphics[width=\linewidth]{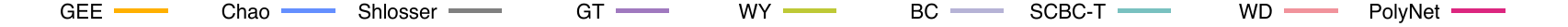}
\end{minipage}\\
\vspace{-0.05in}
\subfigure[\label{fig: Censuspercentage}\scriptsize Census]{\includegraphics[width=0.245\textwidth]{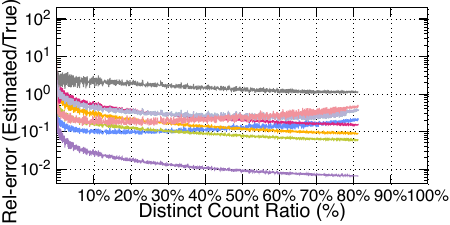}} 
\subfigure[\label{fig: Airlinepercentage}\scriptsize Airline]{\includegraphics[width=0.245\textwidth]{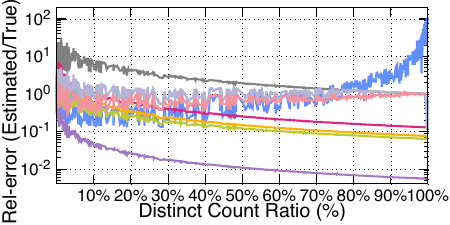}}
\subfigure[\label{fig: DMVpercentage}\scriptsize DMV]{\includegraphics[width=0.245\textwidth]{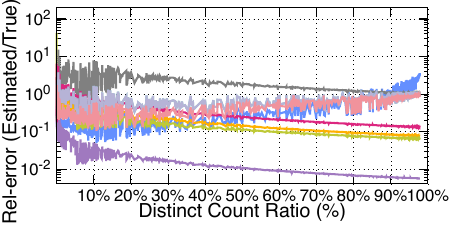}}
\subfigure[\label{fig: Campaignpercentage}\scriptsize Campaign]{\includegraphics[width=0.245\textwidth]{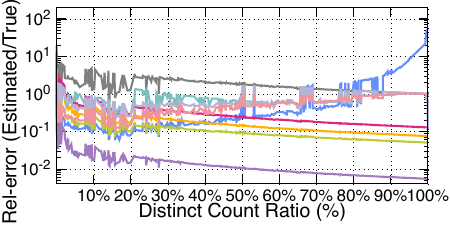}} 
\vspace{-0.15in}
\caption{\label{fig:errorpercentage}Rel-errors on varying distinct count ratio ($D/N$).}
\vspace{-0.05in}
\end{figure*}

\begin{figure*}[ht!]
\centering
\begin{minipage}[]{0.8\textwidth}
\includegraphics[width=\linewidth]{figures/key.pdf}
\end{minipage}\\
\vspace{-0.05in}
\subfigure[\label{fig: census(a)}\scriptsize Census (Queries sorted by $D$)]{\includegraphics[width=0.245\linewidth]{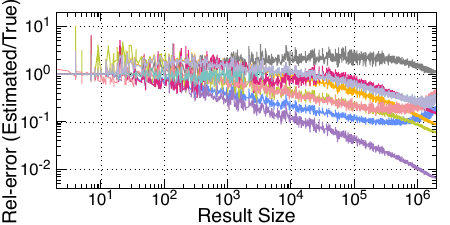}}
\subfigure[\label{fig: census(b)}\scriptsize Census (Queries sorted by number of attributes and distinct count)]{\includegraphics[width=0.745\linewidth]{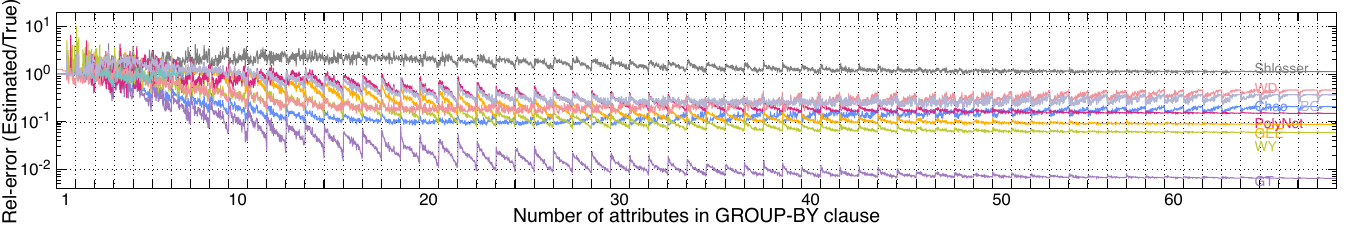}}
\vspace{-0.15in}
\caption{\label{fig:errorperquery}Rel-errors on varying distinct count and number of attributes.}
\vspace{-0.15in}
\end{figure*}

\subsection{Query-level Error Analysis}\label{sec:5.2}
In this section, we delve into the query level to analyze the sources of estimation errors

\underline{\textbf{The impact of the singleton elements.}} As discussed in Section \ref{sec:4.1}, an increase in the number of distinct combinations leads to more unseen elements in the sample, which in turn influences the estimation error. To uniformly evaluate the impact of the distinct count over different datasets, we compute the ratio of the distinct count ($D$) to the total number of records ($N$) in each dataset. We then sort all queries based on their $D/N$ ratio in ascending order. Specifically, on the Census dataset, due to the relatively low number of distinct values per column, the maximum $D/N$ ratio across all queries is 81.2\%, for the other three datasets are 99.9\%, 97.8\%, and 100\%, respectively. The corresponding rel-errors are shown in Figure \ref{fig:errorpercentage}, the observations are as follows:



\textbf{O4: The confidence and confidence interval of multi-attribute distinct value estimation are related to the number of singletons.}
As $D/N$ increases, the number of singleton elements $F_1$ also rises, and reaches $F_1 = N$ when $D = N$. As shown in Figure \ref{fig:errorpercentage}, when there are fewer singletons, all estimators are relatively more accurate, corresponding to narrower confidence intervals. As the number of singletons increases, \textcolor{black}{the analysis needs to be divided into two categories:}
\begin{itemize} [topsep=0pt,leftmargin=12pt]
    \item \textcolor{black}{Methods without boundary information} (GEE, GT, WY, and PolyNet) assume linear relationships between \(D\) and \(f_i\), leading to systematic underestimation as \(F_1\) increases, while the nonlinear Chao estimator tends to overestimate. These estimators’ confidence intervals become wider.
    \item \textcolor{black}{Methods incorporating boundary information (Shlosser, WD, BC, and SCBC-T)} remain more stable and accurate under high singleton regimes. This robustness stems from explicitly encoding population size or theoretically derived bounds into the estimation process.
\end{itemize}



In summary, the confidence intervals for sampling-based estimators are closely related to the number of singletons, and the error trends depend on whether boundary information is considered. This observation provides a foundation for further research into quantifying estimator uncertainty.

\textbf{O5: The bounds of SCBC-T are relatively loose. For the vast majority of queries, the corrected results using single-attribute information are identical to the uncorrected ones. More effective use of single-attribute information offers potential for further improving accuracy.}
Recall Table \ref{tbl:n} and Figure \ref{fig: multiple}, SCBC-T, which corrects boundaries using single-attribute information, achieves higher overall accuracy than BC. 
\textcolor{black}{However, as shown in Figure \ref{fig:errorpercentage}, the curve for SCBC-T is nearly indistinguishable from that of BC in most cases. This overlap occurs because SCBC-T frequently degenerates into the BC estimator when boundary-based constraints are not sufficiently tight. Specifically, across the four datasets, 92\%, 89\%, 91\%, and 36\% of queries, respectively.
This is due to the relatively loose bounds obtained by SCBC-T from single-attribute information.} 
Recalling Section~\ref{sec:sampling-based estimators}, SCBC-T represents the upper bound of SCBC's accuracy by isolating the errors introduced by sketches. It adjusts BC's bounds using true single-attribute information: the number of singleton elements $F_{1,j}$ and the number of distinct values $D_j$ for each attribute $j$. However, for the upper bound $\hat{U}_{SCBC-T} = \min ( \hat{U}_{BC}, \prod_{j=1}^{C}D_j)$. When $\prod_{j=1}^{C}D_j\ge N$, the bound computed from single-attribute distinct values becomes ineffective, and $\hat{U}_{SCBC-T} = \hat{U}_{BC}$. For the lower bound, $\hat{L}_{SCBC-T} = \max( \hat{L}_{BC}, \max(F_{1,j})_{j=1,\ldots, C} )$, the effect of tightening BC's lower bound is only significant when the GROUP-BY involves attributes with a very large number of singleton elements. Consequently, SCBC-T estimates often fall back to BC's estimates. In conclusion, these observations indicate that exploring more effective ways to leverage single-attribute information remains a promising research direction.

\underline{\textbf{The impact of the number of attributes.}}
To isolate the impact of attribute count from distinct count, Figure \ref{fig: census(a)} sorts queries by distinct count, while Figure \ref{fig: census(b)} sorts them by attribute count and then by distinct count. In Figure \ref{fig: census(b)}, dashed lines delineate segments (e.g., 1 to 10) representing the number of attributes involved. We report results for Census as other datasets exhibit similar trends. Comparing these figures, we observe:

\textbf{O6: Incorporating information about the number of attributes also offers potential for multi-attribute distinct combination estimation.} As shown in \autoref{fig:errorperquery}, the overall shape of the curve in Figure \ref{fig: census(b)} closely resembles that in Figure \ref{fig: census(a)}. Taking WD as an example, when the number of attributes is $10$, its error curve shows a downward trend as the distinct count increases, suggesting growing underestimation. However, when the number of attributes exceeds $30$, the error curve of WD exhibits an upward trend, indicating decreasing error. This indicates that as the number of attributes increases, the average number of distinct values also tends to increase. In other words, the number of attributes influences the range of distinct count, thereby affecting estimation accuracy. 

\begin{figure*}[ht!]
\centering
\vspace{-0.05in}
\subfigure[\label{fig: Census-f-1}\scriptsize Census]{\includegraphics[width=0.245\textwidth]{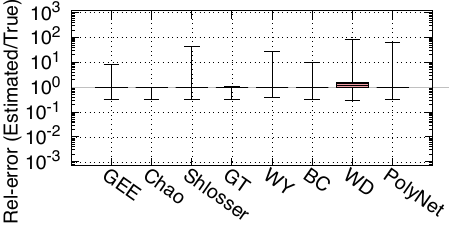}} 
\subfigure[\label{fig: Airline-f-1}\scriptsize Airline]{\includegraphics[width=0.245\textwidth]{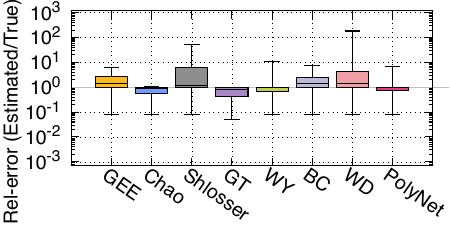}}
\subfigure[\label{fig: DMV-f-1}\scriptsize DMV]{\includegraphics[width=0.245\textwidth]{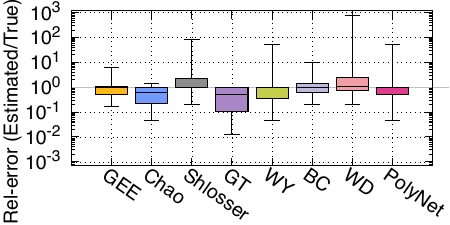}}
\subfigure[\label{fig: Campaign-f-1}\scriptsize Campaign]{\includegraphics[width=0.245\textwidth]{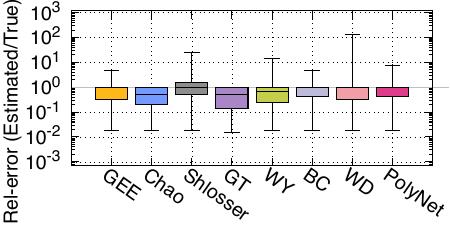}} 
\vspace{-0.15in}
\caption{\label{fig: single-f}Rel-error distribution on single-attribute filtered GROUP-BY queries.}
\vspace{-0.05in}
\end{figure*}

\begin{figure*}[ht!]
\centering
\vspace{-0.05in}
\subfigure[\label{fig: Census-f-n}\scriptsize Census]{\includegraphics[width=0.245\textwidth]{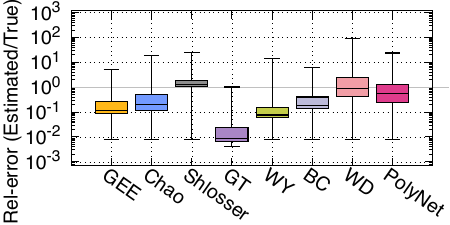}} 
\subfigure[\label{fig: Airline-f-n}\scriptsize Airline]{\includegraphics[width=0.245\textwidth]{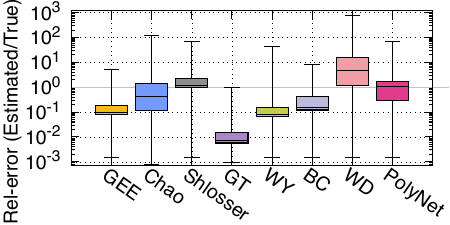}}
\subfigure[\label{fig: DMV-f-n}\scriptsize DMV]{\includegraphics[width=0.245\textwidth]{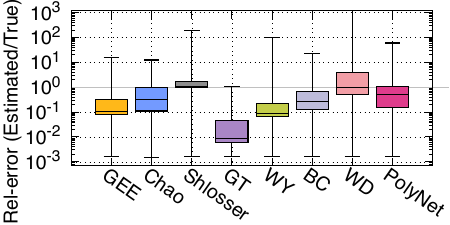}}
\subfigure[\label{fig: Campaign-f-n}\scriptsize Campaign]{\includegraphics[width=0.245\textwidth]{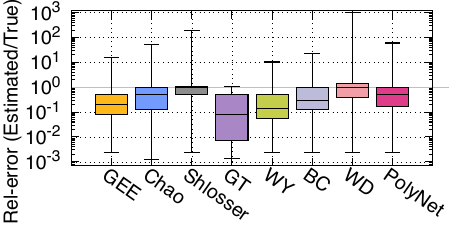}} 
\vspace{-0.15in}
\caption{\label{fig: multiple-f}Rel-error distribution on multi-attribute filtered GROUP-BY queries.}
\vspace{-0.05in}
\end{figure*}

\section{Can Existing Estimators Handle Filtered GROUP-BY Queries?}\label{sec:6}

In this section, we extend our evaluation to multi-attribute GROUP-BY cardinality estimation involving selection predicates. \textcolor{black}{In practical query optimization scenarios, cardinality estimation is subject to strict latency constraints and is typically performed using pre-computed statistics and samples (e.g., the \texttt{ANALYZE} mechanism in PostgreSQL). Accordingly, Sections \ref{sec:6.1} and \ref{sec:6.2} evaluate filtered multi-attribute GROUP-BY queries under the traditional pre-sampled regime, in which a fixed-rate sample is built in advance, and filter predicates are applied on the sample during estimation. In addition, Section \ref{sec:6.3} further explores the Independent Query Sampling (IQS) scenario \cite{iqs} , which represents the potential accuracy gains when the system can afford a more sophisticated, query-dependent sampling process.}

\subsection{How Do Filter Predicates Affect Queries? }\label{sec:6.1}
We present the boxplots depicting the rel-error distribution for the filtered workload at a sampling rate of \( r = 0.005 \) in Figures \ref{fig: single-f} and \ref{fig: multiple-f}, respectively. Recall Section \ref{sec:sampling-based estimators}, the SCBC-T method is a variant of SCBC, used to isolate the error introduced by sketch-based estimation of single-attribute information, representing the theoretical optimum of SCBC. However, SCBC does not support filter predicates, as the original sketch structure cannot estimate the single-attribute distinct values and $F_1$ values under filter predicates. Therefore, in this section, we report the performance of the remaining eight estimators under filtered GROUP-BY queries. The key observations are summarized as follows:

\textcolor{black}{\textbf{O7: Under pre-computed sampling, filter predicates amplify worst-case estimation errors by reducing the effective sample size, with a more severe impact in multi-attribute queries.}}
\textcolor{black}{Figures \ref{fig: single-f} and \ref{fig: multiple-f} evaluate the pre-sampled regime, where filter predicates are applied on top of a fixed-rate sample.} While the median error remains relatively stable, the maximum error increases dramatically due to the sharp reduction in the effective sample size caused by selective filters. For single-attribute queries, filters increase the maximum estimation error by an average of $14\times$, $24\times$, $65\times$, and $10\times$ across the four datasets. \textcolor{black}{For multi-attribute queries, the same effect is further amplified:} comparing Figures \ref{fig: multiple} and \ref{fig: multiple-f} shows average increases of $11\times$, $51\times$, $123\times$, and $93\times$, respectively.  \textcolor{black}{This amplification arises because the reduction in effective sample size affects both settings, but leads to more extreme sparsity in multi-attribute GROUP-BY due to a larger distinct combination space. }

\textbf{O8: Estimating distinct values for queries with filter predicates requires sampling rate information to distinguish between the actual sample size and the number of tuples within the sample that satisfy the predicates.} As shown in Figure \ref{fig: multiple-f}, the error of WD increases significantly compared to Figure \ref{fig: multiple}. This is because WD lacks information about the sampling rate. As Equation (\ref{EQ:WD}), its input includes only the sample profile and the total population size $N$. The sample size $n$ can be implicitly derived as $n = \sum_{i=1}^{m} i \cdot f_i$. However, when filter predicates are introduced, the sample profile only reflects the subset of sampled records that satisfy the predicate, and no longer corresponds to the actual sample size. Hence, when filter predicates are involved, estimators must be able to distinguish between the actual sample size and the number of tuples within the sample that satisfy the predicates.

\vspace{-0.08in}
\subsection{When Are Filtered GROUP-BY Queries More Difficult?}\label{sec:6.2}
\textcolor{black}{In this section, we conduct a query-level analysis to examine the sources of estimation error under the pre-computed sampling regime. In this setting, estimation errors are primarily driven by the effective sample size reduction caused by selective filters, as well as the correlation between the filter attribute and the GROUP-BY attributes.}
\begin{figure}[t!]
\subfigure[\label{fig: category}\scriptsize Percentage]{\includegraphics[width=0.24\linewidth]{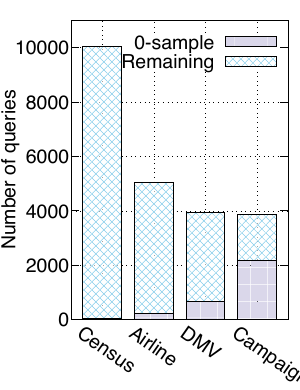}}
\subfigure[\label{fig: extreme-sel}\scriptsize Selectivity]{\includegraphics[width=0.24\linewidth]{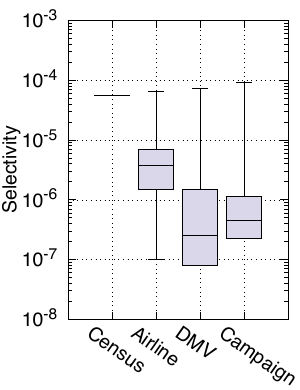}}
\subfigure[\label{fig: extreme-1}\scriptsize $\hat{D_\theta}=1$]{\includegraphics[width=0.24\linewidth]{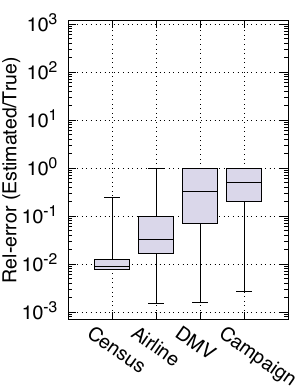}}
\subfigure[\label{fig: extreme-heuristic} \scriptsize Heuristic]{\includegraphics[width=0.24\linewidth]{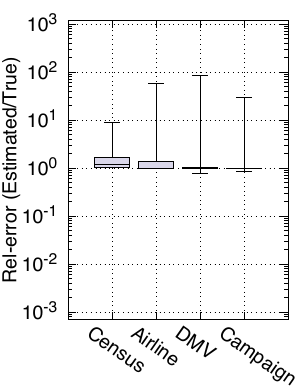}}
\vspace{-0.15in}
\caption{\label{fig:extreme} Analysis of queries of the 0-sample case.}
\vspace{-0.2in}
\end{figure}

\begin{figure*}[ht!]
\centering
\begin{minipage}[]{0.8\textwidth}
\includegraphics[width=\linewidth]{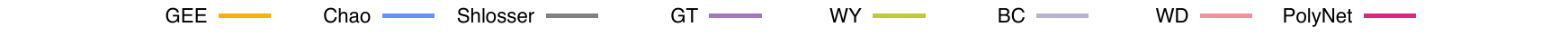}
\end{minipage}\\
\vspace{-0.08in}
\subfigure[\label{fig: census-ratio}\scriptsize Census]{\includegraphics[width=0.49\linewidth]{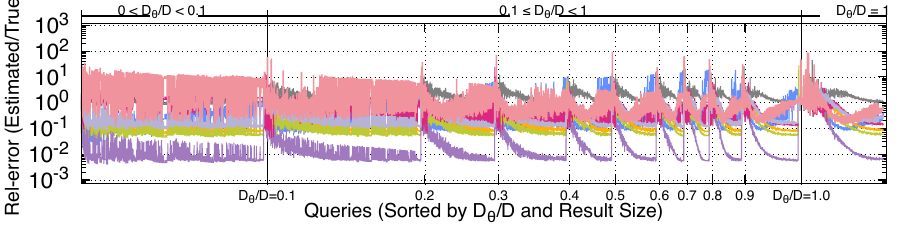}}
\vspace{-0.08in}
\subfigure[\label{fig: airline-ratio}\scriptsize Airline]{\includegraphics[width=0.49\linewidth]{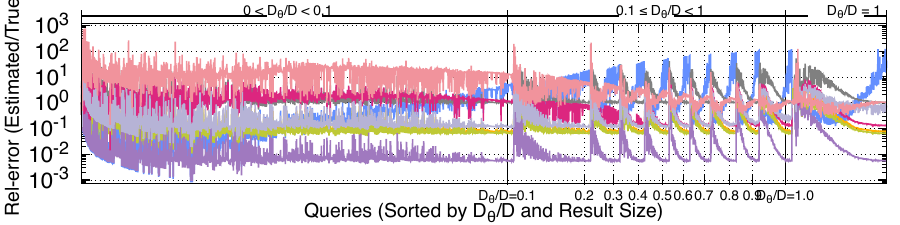}}
\vspace{-0.08in}
\subfigure[\label{fig:dmv-ratio}\scriptsize DMV]{\includegraphics[width=0.49\linewidth]{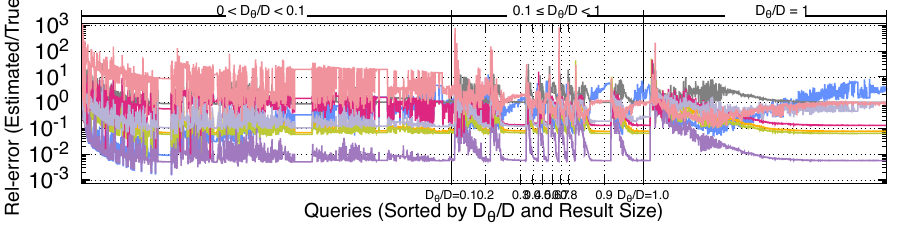}}
\vspace{-0.05in}
\subfigure[\label{fig:campaign-ratio}\scriptsize Campaign]{\includegraphics[width=0.49\linewidth]{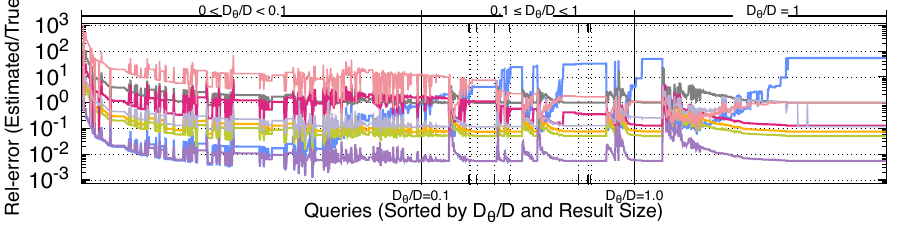}}
\caption{\label{fig:errorperquery-f} Rel-errors on varying filtered distinct count ratio $D_\theta/D$ and query result size $D_\theta$.}
\vspace{-0.1in}
\end{figure*}

We first isolate an extreme case where the selectivity is so high that the sample contains no qualifying tuples. Figure \ref{fig:extreme} analyzes such queries. For the remaining queries, we report the estimation error per query. 
Taking the following query as an example:
\texttt{SELECT} $A_1,A_2$ \texttt{FROM} $\mathcal{R}$ \texttt{WHERE} $\theta_\mathcal{R}$ \texttt{GROUP BY} $A_1,A_2$.
We define:
\begin{itemize} [topsep=0pt,leftmargin=12pt]
    \item Base Distinct Count ($D$): Result size of \texttt{SELECT} $A_1,A_2$ \texttt{FROM} $\mathcal{R}$ \texttt{GROUP BY} $A_1,A_2$.
    \item Filtered Distinct Count ($D_\theta$): Result size after applying filter predicate $\theta$.
\end{itemize}

We compute the filtered distinct count ratio $D_\theta / D$ for each query, rounding it to one decimal place. We then sort the queries based on $D_\theta / D$; for queries with the same ratio, we further sort them by their query result size $D_\theta$. The errors of each query are reported in \autoref{fig:errorperquery-f}. Next, we summarize the key observations as follows:

%
\textcolor{black}{\textbf{O9: Under pre-computed sampling, extremely selective filters may yield empty effective samples; incorporating selectivity information can mitigate this issue beyond using a default estimate.}}
\textcolor{black}{In the pre-sampled setting, empty-sample cases arise when highly selective filter predicates eliminate all tuples from the fixed-rate sample.} Since the number of distinct values cannot be zero, this issue is fundamentally caused by extreme selectivity rather than estimator failure. Figures \ref{fig: category} and \ref{fig: extreme-sel} show that all zero-sample queries across datasets have selectivities below $10^{-4}$, suggesting that selectivity can serve as a reliable signal to adapt estimation strategies.

To illustrate this idea, we introduce a simple heuristic that combines selectivity with unfiltered distinct statistics. Assuming uniform pre-filter group sizes and binomial tuple retention after filtering, the expected number of qualifying groups is estimated as:
\begin{equation} \small
\vspace{-0.05in}
\begin{aligned}
D^{Heuristic}_\theta=D \cdot \left[1 - (1-s)^{N/D}\right]
\label{EQ:heuristic}
\end{aligned}
\end{equation}
where $N$ is the total row count, $D$ is the number of distinct groups before filtering, and $s$ is the selectivity. To isolate the effect of selectivity modeling, we compute this heuristic using exact values of $N$, $D$, and $s$.

Figure \ref{fig: extreme-1} shows that returning a default value of $1$ consistently underestimates the true result in zero-sample cases. In contrast, Figure \ref{fig: extreme-heuristic} demonstrates that the proposed heuristic substantially improves accuracy under extreme selectivity, with over 75\% of queries achieving relative errors close to $1$. These results indicate that, under pre-computed sampling, combining selectivity estimation with unfiltered distinct statistics is a promising direction for mitigating extreme errors in filtered GROUP-BY queries.

\textbf{O10: \textcolor{black}{Under pre-computed sampling,} the main challenge of filtered GROUP-BY estimation stems from the reduction of effective sample size caused by selectivity and correlation.}
Compared with the non-filtered case, estimation errors under filtered queries are larger and less stable. This instability arises because, \textcolor{black}{in the pre-sampled regime,} filter predicates reduce the number of qualifying tuples in the fixed-rate sample, leading to query-dependent effective sample sizes and highly variable estimation behavior. As the ratio $D_\theta / D$ increases, error fluctuations gradually diminish and converge to trends similar to the non-filtered case (Figure \ref{fig:errorpercentage}). The most severe errors occur when $0 < D_\theta / D < 0.1$, where effective samples are particularly sparse. In this regime, effective sample size reduction is driven by two factors: (i) highly selective predicates that retain very few tuples in the pre-sampled data, and (ii) strong correlation between filter predicates and GROUP-BY attributes, which skews the distribution of GROUP-BY keys and further reduces their likelihood of appearing in the sample.

\textcolor{black}{Overall, mitigating errors in filtered GROUP-BY estimation under pre-computed sampling can be approached from two complementary directions: }(1) incorporating additional information on selectivity and correlation (e.g., integrating selectivity estimation into learning-based models), and (2) increasing effective sample sizes, for example through skew-aware or predicate-aware sampling strategies.

\begin{figure*}[ht!]
\centering
\subfigure[\label{fig: Census-iqs-1}\scriptsize  \textcolor{black}{Census}]{\includegraphics[width=0.245\textwidth]{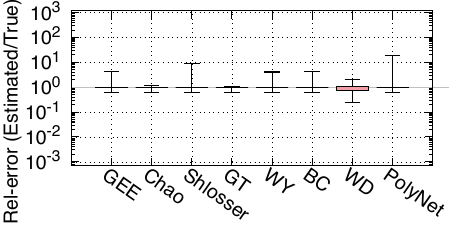}} 
\subfigure[\label{fig: Airline-iqs-1}\scriptsize  \textcolor{black}{Airline}]{\includegraphics[width=0.245\textwidth]{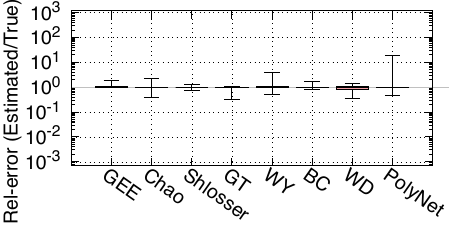}}
\subfigure[\label{fig: DMV-iqs-1}\scriptsize  \textcolor{black}{DMV}]{\includegraphics[width=0.245\textwidth]{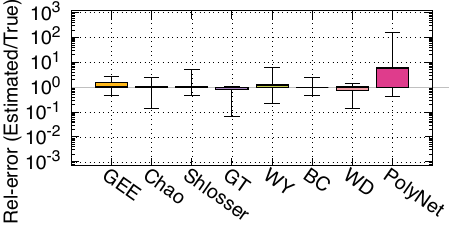}}
\subfigure[\label{fig: Campaign-iqs-1}\scriptsize  \textcolor{black}{Campaign}]{\includegraphics[width=0.245\textwidth]{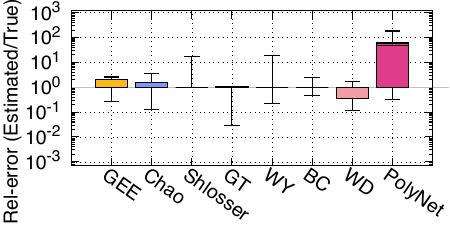}} 
\vspace{-0.12in}
\caption{\label{fig: single-iqs}\textcolor{black}{Rel-error distribution on single-attribute filtered GROUP-BY queries under IQS setting.}}
\vspace{-0.05in}
\end{figure*}

\begin{figure*}[ht!]
\centering
\vspace{-0.05in}
\subfigure[\label{fig: Census-iqs-n}\scriptsize  \textcolor{black}{Census}]{\includegraphics[width=0.245\textwidth]{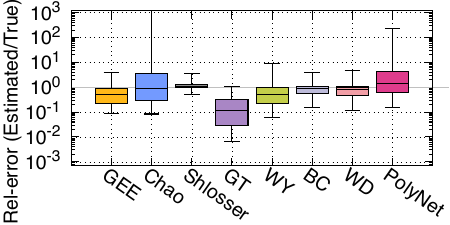}} 
\subfigure[\label{fig: Airline-iqs-n}\scriptsize  \textcolor{black}{Airline}]{\includegraphics[width=0.245\textwidth]{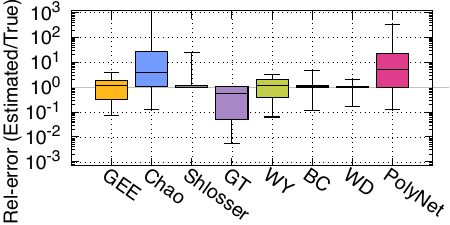}}
\subfigure[\label{fig: DMV-iqs-n}\scriptsize  \textcolor{black}{DMV}]{\includegraphics[width=0.245\textwidth]{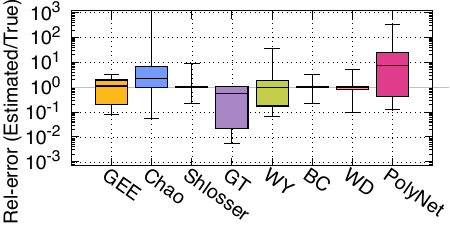}}
\subfigure[\label{fig: Campaign-iqs-n}\scriptsize  \textcolor{black}{Campaign}]{\includegraphics[width=0.245\textwidth]{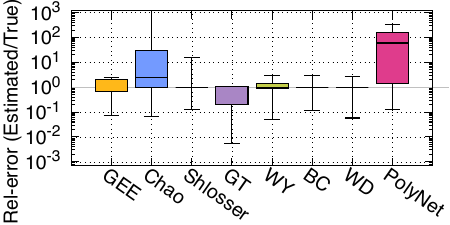}} 
\vspace{-0.12in}
\caption{\label{fig: multiple-iqs}\textcolor{black}{Rel-error distribution on multi-attribute filtered GROUP-BY queries under IQS setting.}}
\vspace{-0.1in}
\end{figure*}

\textbf{O11: The distribution of a single attribute in the GROUP-BY clause affects the overall distribution.} In Figures \ref{fig:dmv-ratio} and \ref{fig:campaign-ratio}, all methods exhibit flat line segments. This occurs because when a particular attribute has a significantly high distinct count, it tends to dominate the overall distribution. As a result, whenever a GROUP-BY query includes this attribute, the query's cardinality approximates the distinct count of that attribute. Additionally, the attribute's distribution also dominates the sample profile. Therefore, different queries involving this dominant attribute will produce similar sample profiles, leading to similar estimation errors. This insight suggests that incorporating information about individual attributes could improve the estimation of GROUP-BY queries.

\vspace{-0.05in}
\subsection{\textcolor{black}{Evaluation on Independent Query Sampling Scenario}} \label{sec:6.3}

\textcolor{black}{In this section, we report the relative errors for single-attribute and multi-attribute distinct estimation under the IQS setting in Figures \ref{fig: single-iqs} and \ref{fig: multiple-iqs}. Unlike pre-sampling, Independent Query Sampling (IQS) aims to construct mutually independent samples for different queries, thereby preventing error accumulation and instability caused by correlated samples across repeated queries \cite{iqs}. The key observation is as follows:}


\textcolor{black}{\textbf{O12: IQS enhances estimation stability by mitigating empty-sample issues, yet its accuracy remains constrained by the inherent population size under extreme selectivity, with varying robustness across learning-based models.}} \textcolor{black}{As shown in Figures \ref{fig: single-iqs} and \ref{fig: multiple-iqs}, IQS yields more stable error distributions. However, a fundamental limitation persists: when predicate selectivity is so high that the number of qualifying tuples is smaller than the target sample size, the reduced effective sample size inevitably lowers accuracy compared to non-filtered settings. Among the evaluated methods, PolyNet is the most sensitive to selectivity-induced changes in sample size. This is because PolyNet’s polynomial coefficients are learned for a specific sampling-rate range ($0.001$–$0.01$) \cite{li2024learning}, and fail to generalize when high selectivity pushes the effective rate below this threshold. In contrast, WD approximates a maximum likelihood estimator (MLE), which remains statistically valid across varying sampling rates, rather than learning an empirically optimal function for a fixed sampling rate. As a result, its estimation error is more stable. Finally, from an efficiency perspective, unlike pre-sampling, which incurs a one-time materialization cost, IQS performs sampling at query time. Prior work \cite{iqs} shows that, with appropriate indexing, the per-query overhead of IQS can be reduced to  \(O(\log N + n)\), yielding more stable estimation at the expense of higher execution overhead.}


\section{\textcolor{black}{Experiments on the TPC-H benchmark}} \label{sec: tpch}

\textcolor{black}{To complement our synthetic benchmarks, we further evaluate the estimators using the TPC-H benchmark. Our evaluation consists of two parts: Section \ref{sec: multitable} studies multi-table scenarios, while Section \ref{sec: plan} analyzes the practical impact of distinct value estimation on query plan selection.}

\vspace{-0.05in}
\subsection{\textcolor{black}{Multi-Table Distinct Value Estimation}}\label{sec: multitable}
\textcolor{black}{For multi-table queries, constructing representative samples over joins is a challenging problem and orthogonal to the focus of this paper. Therefore, we adopt a widely used \emph{correlated sampling} strategy to obtain join samples that approximately reflect the underlying join distribution. For queries containing nested subqueries, we further rewrite them into equivalent join-based formulations that preserve the grouping semantics relevant to distinct value estimation. We report the error distributions for multi-table queries with single-attribute and multi-attribute GROUP-BY clauses in Figures~\ref{fig: tpch-1} and~\ref{fig: tpch-n}, respectively. The key observation is as follows:}

\begin{figure}[t!]
\subfigure[\label{fig: tpch-1}\scriptsize  \textcolor{black}{Single Attr.}]{\includegraphics[width=0.48\linewidth]{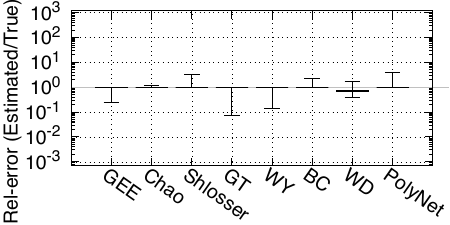}}
\subfigure[\label{fig: tpch-n}\scriptsize  \textcolor{black}{Multi. Attr.}]{\includegraphics[width=0.48\linewidth]{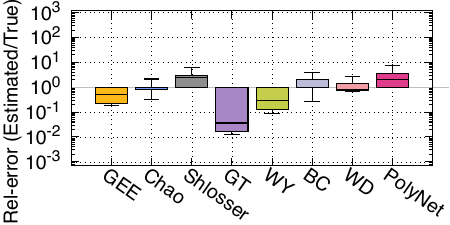}}
\vspace{-0.12in}
\caption{\label{fig:tpcha} \textcolor{black}{Evaluation on the TPC-H benchmark.}}
\vspace{-0.1in}
\end{figure}

\textbf{\textcolor{black}{O13: When join samples faithfully reflect the join output distribution, the conclusions drawn from single-table experiments continue to hold in the multi-table setting.}}
\textcolor{black}{In TPC-H, the schema and workload characteristics allow correlated sampling to produce join samples that reasonably reflect the join distribution. As shown in Figures~\ref{fig: tpch-1} and~\ref{fig: tpch-n}, all estimators achieve near-perfect accuracy for single-attribute queries, while multi-attribute GROUP-BY queries exhibit higher distinct counts and correspondingly larger errors. Nonetheless, the error trends remain consistent with those observed in single-table experiments. Moreover, as estimation is performed on join samples, we empirically observe that the results are largely insensitive to whether the grouping attributes originate from the same table or span multiple tables, provided that the join samples reflect the true join distribution.}

\vspace{-0.05in}
\subsection{\textcolor{black}{Impact of Group-By Cardinality on Query Planning}}\label{sec: plan}

\textcolor{black}{We inject GROUP-BY cardinality estimates into PostgreSQL to study their practical impact on query plans. Figure~\ref{fig: tpch-time} reports the resulting speedup or slowdown relative to the default plans on TPC-H queries. Except for methods with severe underestimation (GT and WY), all other estimators yield performance close to that with true cardinalities, reflecting the generally accurate GROUP-BY estimates in TPC-H (Section~\ref{sec: multitable}).}
\begin{figure}[t!]
\subfigure[\label{fig: tpch-time}\scriptsize \textcolor{black}{Relative time}]{\includegraphics[width=0.231\linewidth]{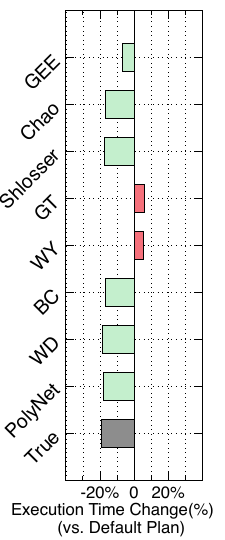}}
\subfigure[\label{fig: tpch-plan1}\scriptsize \textcolor{black}{Parallelism = 1}]{\includegraphics[width=0.23\linewidth]{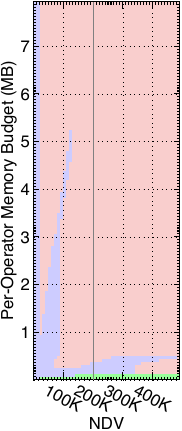}}
\subfigure[\label{fig: tpch-plan2}\scriptsize  \textcolor{black}{Parallelism = 2}]{\includegraphics[width=0.23\linewidth]{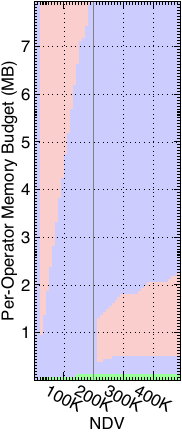}}
\subfigure[\label{fig: tpch-plan3}\scriptsize \textcolor{black}{Parallelism = 3}]{\includegraphics[width=0.278\linewidth]{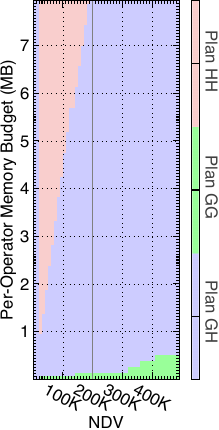}}
\vspace{-0.12in}
\caption{\label{fig:tpchb} \textcolor{black}{Impact of Group-By Cardinality Estimation on Query Planning and Execution}}
\vspace{-0.1in}
\end{figure}


\textcolor{black}{In TPC-H queries, GROUP-BY clauses primarily occur in the top-level query, where the PostgreSQL optimizer typically determines scan and join orders before adding the final aggregation node, meaning GROUP-BY cardinality estimates do not influence join planning in these cases.
PostgreSQL employs two main approaches: \texttt{HashAgg}, which builds an in-memory hash table scaling with the number of groups, and \texttt{GroupAgg}, which sorts tuples by grouping attributes and aggregates sequentially. For parallel execution, the optimizer considers three plans: \emph{Plan GH} (GroupAgg finalize + Partial HashAgg), \emph{Plan GG} (GroupAgg finalize + Partial GroupAgg), and \emph{Plan HH} (HashAgg finalize + Partial HashAgg). 
Using a single-table query on \texttt{lineitem} grouped by \texttt{l\_partkey} as a case study, we vary the NDV estimates (1 to 500K; True NDV = 200K), available memory (\texttt{work\_mem} = 64 to 8192 KB), and parallelism (1 to 3 workers) to map the optimizer's decision boundaries (Figures~\ref{fig: tpch-plan1}-\ref{fig: tpch-plan3}).}

\textcolor{black}{As shown in these figures, the optimizer uses the NDV estimate to project memory requirements: Under extremely limited memory, it tends to select \emph{Plan GG} due to low memory requirements despite higher sorting cost, while with sufficient memory, HashAgg-based plans are preferred. 
Notably, in high-NDV regions, increasing parallelism shifts the choice from hash-based finalization (\emph{Plan HH}) to sort-based finalization (\emph{Plan GH}); this occurs because the finalization phase must construct a global hash table that is more likely to exceed the \texttt{work\_mem} budget as the number of groups increases.}

\textcolor{black}{In summary, overestimating the GROUP-BY cardinality leads to unnecessary and expensive sorting, while underestimation can cause disk spills during hash table construction that degrade performance.
Near the decision boundary between \texttt{HashAgg} and \texttt{GroupAgg}, plan selection is highly sensitive to NDV estimation errors. Furthermore, in more complex workloads like TPC-DS, GROUP-BY frequently appears in subqueries (e.g.,  TPC-DS Q1). Inaccurate NDV estimates in these subqueries propagate upward and mislead the optimizer in selecting join plans at higher levels \cite{freitag2019every}. These findings underscore the critical role of robust NDV estimation for effective query optimization in practice.}

\vspace{-0.05in}
\section{discussion}\label{sec:7}

\underline{\textbf{Main Findings.}}
We conduct a comprehensive evaluation of the multi-attribute distinct value estimation problem.
Based on this evaluation, 
we summarize our key observations \textcolor{black}{and discuss the research questions (\textbf{Q1}–\textbf{Q3}) posed in the introduction.}
\begin{itemize} [topsep=0pt,leftmargin=0pt] 

\item \textcolor{black}{\textbf{On Q1}:}
Estimators relying solely on joint distribution information, including learning-based models, exhibit substantially higher errors in multi-attribute settings (\textbf{O1}) and may introduce systematic over- or underestimation due to heuristic distributional assumptions (\textbf{O3}).

\item  \textcolor{black}{\textbf{On Q2}:} 
The existing method that incorporates single-attribute information tightens bounds only in limited cases and often reverts to the original estimates (\textbf{O5}), indicating that single-attribute statistics are not yet fully exploited.

\item  \textcolor{black}{\textbf{On Q3}:
Existing methods incur larger errors under filtered GROUP-BY queries. This is primarily due to selectivity and correlation effects, which result in sparse qualifying samples (\textbf{O7}). While IQS alleviates this issue, it does not eliminate it entirely (\textbf{O9}, \textbf{O10}, \textbf{O12}). In addition, learning-based models must carefully account for correct sampling rate information and generalize across varying sampling rates (\textbf{O8}, \textbf{O12}). For joins, the primary challenge lies in constructing representative join samples, which is orthogonal to this work (\textbf{O13}).}

\end{itemize}

\noindent\underline{\textbf{Research Opportunities.}} 
Our findings highlight two key directions for future research: 
\begin{itemize} [topsep=0pt,leftmargin=0pt] 

\item \textbf{Improving Estimation Accuracy:} 
\textcolor{black}{Our experiments show that WD achieves stable accuracy by formulating distinct value estimation as a maximum likelihood estimation (MLE) problem,
offering high portability across workloads (\textbf{O2}, \textbf{O12}). For multi-attribute settings, single-attribute marginal distributions provide critical information. As shown by Theorem~1, knowing marginals strictly tightens the lower bound on estimation error by constraining the feasible space of joint distributions; while this does not resolve worst-case indistinguishability, it eliminates a large class of pathological joint distributions. Empirically, this effect is reflected by the clear advantage of SCBC-T on high-NDV datasets. However, SCBC-T applies marginal-derived bounds only as a loose, post hoc correction and often degenerates to the original estimates in most scenarios (\textbf{O5}). Thus, a promising direction is to combine these insights by formulating multi-attribute distinct estimation as a constrained MLE problem, where bounds derived from marginals (e.g., Theorem~1) and empirically relevant factors such as singleton counts, single-attribute statistics, and attribute dimensionality (\textbf{O4}, \textbf{O6}, \textbf{O11}) are integrated directly into the training of learning-based estimators, potentially mitigating systematic biases from heuristic assumptions.}


\item 
\textbf{Quantifying Estimation Uncertainty:} \textcolor{black}{Existing estimators typically provide only point estimates, lacking a measure of confidence. Our observations (\textbf{O4} and \textbf{O8}-\textbf{O12}) reveal that estimation errors are strongly correlated with sample-level statistics, such as the singleton ratio and effective sample size, providing a concrete rationale for uncertainty modeling. Future research can leverage techniques like conformal prediction to learn these correlations and produce rigorous uncertainty intervals. 
As analyzed in Section VII.B, such quantification empowers optimizers to make risk-aware decisions rather than relying on potentially inaccurate point estimates.}
\end{itemize}

\vspace{-0.05in}
\section{Conclusion}\label{sec:9}
This paper presents a comprehensive empirical study of multi-attribute GROUP-BY cardinality estimation. 
Through synthetic workloads and the multi-table TPC-H benchmark, we show that multi-attribute GROUP-BY cardinality estimation remains challenging in practice. Our findings highlight both the limitations of existing approaches and concrete directions toward more robust and practically effective GROUP-BY cardinality estimators.

\vspace{-0.05in}
\section*{Acknowledgment}
The work is supported by National Key R\&D Program of China (2024YFF0617702), NSFC (Nos. U22A2025, 62232007, U23A20309), 111 Project (No. B16009), and the Fundamental Research Funds for the Central Universities (No. N25FZD001).
\section*{
AI-Generated Content Acknowledgement
} \label{sec:9}

This work was conducted entirely by the authors. All ideas, analyses, proofs, figures, and data were manually developed. No AI tools were used to generate or modify any technical content.
\bibliographystyle{IEEEtran}
\bibliography{sample}

\end{document}